\documentclass[aps,pra,twocolumn,a4,nopacs,twoside,floatfix,superscriptaddress]{revtex4-1}
\usepackage{amsmath,graphicx,epsfig,color,amsfonts,amssymb,bbm}
\usepackage{amsthm}
\usepackage{mathtools}
\usepackage{changes}
\usepackage{times,color}
\usepackage{pifont}

\definecolor{myurlcolor}{rgb}{0,0,0.7}
\definecolor{myrefcolor}{rgb}{0.8,0,0}
\usepackage{hyperref}
\usepackage{footnote}
\hypersetup{colorlinks, linkcolor=myrefcolor,
citecolor=myurlcolor, urlcolor=myurlcolor}

\newcommand{\ket}[1]{|#1\rangle}
\newcommand{\bra}[1]{\langle#1|}

\newcommand{\id}{\mathbbm{1}}

\newcommand{\B}{\mathcal  B}

\definecolor{nblue}{rgb}{0.2,0.2,0.7}
\definecolor{ngreen}{rgb}{0.2,0.6,0.2}
\definecolor{nred}{rgb}{0.8,0.2,0.2}
\definecolor{nblack}{rgb}{0,0,0}

\newtheorem{theorem}{Theorem}

\begin{document}

\title{A versatile construction of Bell inequalities for the multipartite scenario}
%:\\A Single Family of Bell inequalities to detect genuine multipartite nonlocality of pure states}

\author{Florian J. Curchod}
\email{Florian.Curchod@icfo.es}
\affiliation{ICFO-Institut de Ciencies Fotoniques, The Barcelona Institute of Science and Technology, 08860 Castelldefels (Barcelona), Spain}

%\affiliation{?}

%\affiliation{  }

\author{Mafalda L. Almeida}
\affiliation{ICFO-Institut de Ciencies Fotoniques, The Barcelona Institute of Science and Technology, 08860 Castelldefels (Barcelona), Spain}

\author{Antonio Ac\'{i}n}
\affiliation{ICFO-Institut de Ciencies Fotoniques, The Barcelona Institute of Science and Technology, 08860 Castelldefels (Barcelona), Spain}
\affiliation{ICREA--Instituci\'{o} Catalana de Recerca i Estudis Avan\c{c}ats, E--08010 Barcelona, Spain}

%\affil[1]{ICFO-Institut de Ciencies Fotoniques, The Barcelona Institute of Science and Technology, 08860 Castelldefels (Barcelona), Spain}
%\affil[2]{ICFO-Institut de Ciencies Fotoniques, The Barcelona Institute of Science and Technology, 08860 Castelldefels (Barcelona), Spain}
%\affil[3]{ICREA--Instituci\'{o} Catalana de Recerca i Estudis Avan\c{c}ats, E--08010 Barcelona, Spain}

%\maketitle

\date{\today}

\begin{abstract}
Local measurements acting on entangled quantum states give rise to a rich correlation structure in the multipartite scenario. %}{The multipartite scenario offers a large and complex variety of nonlocal correlations.}\deleted{Here w}
We introduce a versatile technique to build families of Bell inequalities witnessing different notions of multipartite nonlocality for any number of parties. The idea behind our method is simple: a known Bell inequality satisfying certain constraints, for example the Clauser-Horne-Shimony-Holt inequality, serves as the \emph{seed} to build new families of inequalities for more parties. The constructed inequalities have a clear operational meaning, capturing an essential feature of multipartite correlations: their violation implies that numerous subgroups of parties violate the inequality chosen as seed. %, conditioned on the remaining parties' inputs and outputs. 
The more multipartite nonlocal the correlations, the more subgroups can violate the seed. We illustrate our construction using different seeds and designing Bell inequalities to detect $m$-way nonlocal multipartite correlations, in particular, \emph{genuine multipartite nonlocal} correlations -- the strongest notion of multipartite nonlocality. For one of our inequalities we prove analytically that a large class of pure states that are genuine multipartite entangled exhibit genuine multipartite nonlocality for any number of parties, even for some states that are almost product. We also provide numerical evidence that this family is violated by all genuine multipartite entangled pure states of three and four qubits. Our results make us conjecture that this family of Bell inequalities can be used to prove the equivalence between genuine multipartite pure-state entanglement and nonlocality for any number of parties.%, obtaining a genuine multipartite version of Gisin's theorem.

%Moreover, we extend from three to four-parties the numerical evidence that this family of inequalities is sufficient to detect GMNL in all GME pure qubit states. We conjecture that this family of Bell inequalities allows us to prove that all GME pure states are GMNL, obtaining a genuine multipartite version of Gisin's theorem.

% Its construction is inspired by a property of genuine entangled pure states: with these only can any two of the observers be projected into a pure entangled state by a suitable projection of the other observers.
% and corresponds to the situation where \textit{all} the observers are behaving nonlocaly in a global way%The measurement observer $A_1$ (resp. $A_2$) choses to make is labeled by $x_1 (x_2) \in \{0,1\}$ and the outcome $a_1 (a_2) \in \{0,1\}$, both are dichotomic. 

\end{abstract}

\maketitle

\section{Introduction}
Quantum theory is rich in features that defy classical intuition. Systems of several particles are particularly interesting in that sense, with quantum systems exhibiting more intricate correlations than those possible within classical ones. For instance, some composite quantum systems can not be specified by the state of their parts alone, but require a global description -- a phenomenon known as \textit{quantum entanglement}. When the parts of such entangled systems are separated at a distance, and local (independent) measurements are made on them, the distribution of outcomes can exhibit \textit{nonlocal correlations}, in the sense that they cannot be explained by the existence of a (possibly hidden) classical common cause~\cite{EPR,Bell}. Apart from their fundamental interest, quantum entanglement and quantum nonlocality have been identified as key resources for Quantum Information Science. 

Nonlocal correlations have been extensively studied in the simplest scenario of bipartite systems, which is sufficient to obtain powerful resources for information tasks with no classical equivalent: %an efficient factorisation algorithm~\cite{Shor}, 
randomness expansion~\cite{colbeck2009quantum, ColbeckExpansion, Randomness} and amplification~\cite{ColbeckAmplification}, distribution of secret keys in a provably secure way~\cite{EkertQKD, barrett2005no, acin2007device} or testing the functioning of devices with minimal assumptions on their internal machinery~\cite{MayersSelfTesting}, for example. 

Multipartite scenarios -- consisting of set-ups with at least three parties --  have received far less attention due to their greater complexity. They offer, however, a much richer source of correlations than the bipartite set-up, and have already been proven useful for several tasks. 
%tasks as randomness amplification under a minimal set of assumptions~\cite{??}, detection of the non-local behavior of ground-states of Hamiltonians appearing naturally in nuclear physics~\cite{??}, and exclusion of causal influences that spread at (finite) speed greater than light ~\cite{PironioSpeed}.
Either for a better use of the potential provided by multipartite systems -- which might be particularly interesting for tasks on quantum networks -- or simply to explore scenarios that go beyond the standard bipartite set-up, the study of multipartite scenarios is nowadays a central problem~\cite{Svetlichny,Svetlichny2,Gallego:PRL:070401,Bancal:PRA:014102,BancalDIEW,gallego2013full,bancal2014quantum,BoudaRandomness,Jordi}.%\\

A detailed study of correlations in the multipartite scenario is an increasingly demanding task, as the complexity of the possible states of the systems and sets of correlations grows exponentially with the number of parties. %, given that the number of inequivalent possible states of a system grows exponentially with its number of parties. 
A common approach to characterise multipartite correlations consists of testing whether they can be reproduced by models in which the parties share different physical resources (classical, quantum or post-quantum correlations)~\cite{Svetlichny,BancalQuantifying,Bancal:PRA:014102,Gallego:PRL:070401}. These models range from completely classical -- where all parties can only share classical correlations, to \emph{genuine multipartite}  -- where all parties are required to be non-classically correlated. Intermediate models include, for instance, hybrid models where non-classical correlations are allowed inside groups of the parties, but the different groups can only be correlated classically between each other %the system is partitioned into subsystems, which can only share classical correlations while non-classical correlations are allowed inside the subsystems 
~\cite{Bancal:PRA:014102}. Although families of Bell inequalities that provide insight on the rich structure of multipartite scenario have been built \cite{Svetlichny,BancalQuantifying,Bancal:PRA:014102}, we are far from a complete characterisation of multipartite correlations.

Understanding the precise relation between entanglement and nonlocality in the multipartite scenario is of particular interest.  While quantum entanglement is necessary for the display of quantum nonlocality, it is not sufficient. Indeed, there exist entangled mixed states for which single local measurements never generate nonlocal correlations \cite{Werner,BarrettPOVM}. %,Barrett,Almeida,TothAcin}. 
%%In some cases, nonlocality can be extracted from such states by considering generalised scenarios, which include the use of sequences of local measurements~\cite{PopescuHidden} or collective local measurements on many copies of the states~\cite{Palazuelos}. 
Remarkably, bipartite systems in a pure quantum state display a straighforward relation between entanglement and nonlocality: all pure entangled bipartite states are nonlocal, a result known as Gisin's Theorem ~\cite{GisinThrm}. This result has been extended to the multipartite scenario ~\cite{PopescuGeneric,PopescuErratum}, with the caveat that the used definitions of entanglement and nonlocality do not capture any truly multipartite features (with these definitions, a multipartite system is said to be non-classically correlated if at least two parties share non-classical correlations). Partial results have been obtained for the genuine multipartite (GM) notions of entanglement and nonlocality: all three-qubit systems in a GM entangled (GME) pure state are GM nonlocal (GMNL), as well as any $n$-qubit systems in a fully-symmetric GME pure state~\cite{Yu3,Chinese1}.  However, these results rely on the use of GM Hardy-type paradoxes ~\cite{Hardy}, which have the drawback of not allowing for experimental tests, contrary to nonlocal correlations detected by the violation of a Bell inequality.

%In order to grasp the full potential of many-body systems, it is necessary to consider genuinely multipartite definitions of entanglement and nonlocality, where all the parties of a system are engaged, instead of only subsets of them.
%\emph{Genuine multipartite entangled (GME)} states are necessary to generate \emph{genuine multipartite nonlocal (GMNL)} correlations. It is a longstanding open question whether entanglement  in pure states is sufficient to observe nonlocality, in the genuinely multipartite sense. Can all pure GME states generate GMNL correlations, for any number of particles?
%So far, it is known that this holds for systems of three particles in a pure GME state~\cite{Yu3,Chinese1}. This result relies, however, on the use of genuine tripartite Hardy-type paradoxes ~\cite{Hardy}, which have the drawback of not allowing for experimental realisations, contrary to nonlocal correlations detected by the violation of a Bell inequality \footnote{Contrary to the violation of a Bell inequality, the realisation of a Hardy-type paradox relies on strong conditions of the form $P(ab|xy)=0$ for some values of $a,b,x,y$. Such conditions are impossible to meet in an experiment, where even the smallest imperfections lead to values $P(ab|xy)= \epsilon > 0$. Realisations that are close to the optimal ones of a Hardy paradox are likely to remain nonlocal, but will further require the use of a Bell inequality as witness. It is moreover unclear which Bell inequality should be used in that case.}.\\
%
%
In this work we introduce a new technique to build Bell inequalities for the detection of truly multipartite nonlocal correlations, in a no-signalling framework \cite{Bancal:PRA:014102}. These inequalities have a very clear operational meaning and capture essential features of multipartite nonlocality. Our construction takes a ``seed'' -- a Bell inequality that fulfils certain constraints -- to generate Bell inequalities for an arbitrary number of parties. The inequalities can be designed to detect $m$-way nonlocality, for any $2 \leq m \leq n$, including the extreme case of GMNL ($m=2$).  We illustrate the potential of the method by constructing several families of multipartite Bell inequalities from different seeds and for different notions of multipartite nonlocality.

Our technique is particularly fit for the detection of multipartite nonlocal correlations of pure states. Indeed, using the CHSH inequality~\cite{CHSH} as the seed, we design two families of Bell inequalities, $I_{\textrm{sym}}$ and $I_{\textrm{\ding{192}}}$ that detect GMNL in large classes of GME pure states. Note that the CHSH inequality has already been used to prove the equivalence between pure state entanglement and nonlocality for bipartite systems~\cite{GisinThrm}. Moreover, for three parties, $I_{\textrm{sym}}$ coincides with a Bell inequality obtained in \cite{Bancal:PRA:014102}, for which the authors found numerical evidence that the equivalence holds for all three-qubit states. Here we show analytically that, for any number of parties, all pure GHZ-like states that are GME contain GMNL correlations detected by $I_{\textrm{sym}}$, even almost separable states. We supplement these analytical results by providing numerical evidence that all four-qubit systems in a GME pure state violate $I_{\textrm{sym}}$. In the tripartite scenario, using $I_{\textrm{\ding{192}}}$, we also show analytically  that all pure states symmetrical under the permutation of two parties are GMNL.
The partial results obtained added to the operational meaning of our construction lead us to conjecture that  the family of Bell inequalities $I_{\textrm{sym}}$ can be used to generalise Gisin's theorem, proving that all GME pure states are GMNL.  \\

\section{The tripartite scenario}
We start by introducing the main concepts used in this work and our results for the tripartite scenario, which is the simplest multipartite scenario for the observation of nonlocal correlations. 
This scenario counts with three distant observers $A_i$, $i\in\{1,2,3\}$ making rounds of measurements on multipartite quantum systems. At each round, the choice of local measurement performed by each party is labelled $x_i$ and the obtained outcome $a_i$. %The choice of measurement is made locally and the outcomes define space-like separated events. 
The generated joint conditional probability distribution $P(a_1a_2a_3|x_1x_2x_3)$ is then said to be \textit{local} if it factorises, given the additional knowledge of a (possibly hidden) common classical cause $\lambda$:
\begin{equation}\begin{split}\label{FullyLocal}
P_{\mathcal{L}}(a_1a_2a_3|x_1x_2x_3) \\ =\sum\limits_{\lambda} q(\lambda) P_{A_1}(a_1|x_1,\lambda) P_{A_2}(a_2|x_2,\lambda)& P_{A_3}(a_3|x_3,\lambda)
\end{split}\end{equation}
The common cause $\lambda$  is a discrete random variable with distribution $q(\lambda) \geq 0$, $\sum\limits_{\lambda} q(\lambda) = 1$, and $P_{A_i}(a_i|x_i,\lambda)$  is a probability distribution for party $A_i$. A distribution $P(a_1a_2a_3|x_1x_2x_3)$ that does not allow for a decomposition \eqref{FullyLocal} is said to be \emph{nonlocal}. Note that this definition of locality for three parties is a straightforward generalisation of the bipartite scenario, where the only difference is the addition of a third party. 
Because of the measurement arrangements it is assumed that the \emph{no-signalling (NS)} principle~\cite{GhirardiNS} holds, i.e. party $A_1$ cannot signal to the other parties by performing a choice of measurement
\begin{equation}\begin{split}\label{NScond}
P(a_2a_3|x_2x_3) \equiv P(a_2a_3|x_1x_2x_3)\\=\sum_{a_1}P(a_1a_2a_3|x_1x_2x_3), \hspace{0.2cm} \forall x_1
\end{split}\end{equation}
and similarly for parties $A_2$ and $A_3$.

The notion of separability for a tripartite pure state $\ket{\psi_{123}}$ is also a direct extension of the bipartite case, $\ket{\psi_{123}}=\ket{\phi_1}\ket{\phi_2}\ket{\phi_3}$, where $\ket{\phi_{i}}$ is the state of party $A_i$. The state $\ket{\psi_{123}}$ is then entangled whenever it does not admit for the previous decomposition. In that case, it is already known to be nonlocal, as there always exist local measurements on it that lead to a nonlocal joint distribution \cite{PopescuGeneric}. This equivalence between pure state entanglement and nonlocality is however essentially the same as for bipartite systems \cite{GisinThrm}, since it only requires two parties to be entangled.\\

Here we are interested in genuinely multipartite definitions of entanglement and nonlocality. As first noticed by Svetlichny~\cite{Svetlichny}, distributions generated in a tripartite scenario lead to stronger notions of nonlocality. Consider for instance a relaxation of the locality assumption, where pairs of parties are now allowed to group together and share nonlocal resources. This type of hybrid local/nonlocal models leads to joint conditional probability distributions 
\begin{equation}\begin{split}\label{GenuineNonlocal}
P_{2/1}(a_1a_2a_3|x_1x_2x_3)& = \\ = \sum\limits_{\lambda_1} q_1(\lambda_1) P_{A_1A_2}(a_1a_2|x_1x_2,\lambda_1) & P_{A_3}(a_3|x_3,\lambda_1)\\ + \sum\limits_{\lambda_2} q_2(\lambda_2) P_{A_1A_3}(a_1a_3|x_1x_3,\lambda_2) & P_{A_2}(a_2|x_2,\lambda_2)\\+ \sum\limits_{\lambda_3} q_3(\lambda_3)  P_{A_2A_3}(a_2a_3|x_2x_3,\lambda_3) & P_{A_1}(a_1|x_1,\lambda_3)
\end{split}\end{equation}
with $q_i(\lambda_i) \geq 0$ and $\sum\limits_{i_,\lambda_i} q_i(\lambda_i)= 1$. Distributions $P(a_1a_2a_3|x_1x_2x_3)$ that cannot be decomposed in the form \eqref{GenuineNonlocal} are named \emph{genuine tripartite nonlocal}. As shown in \cite{Gallego:PRL:070401}, the original notion of multipartite nonlocality by Svetlichny \cite{Svetlichny} faces operational problems. To avoid these, one assumes that the no-signalling principle~\cite{NS} also holds at the level of distributions $P_{A_iA_j}(a_ia_j|x_ix_j,\lambda)$, which implies that the marginals $P(a_i|x_i,\lambda) = P(a_i|x_ix_j,\lambda) = \sum_{a_j} P(a_ia_j|x_ix_j,\lambda), \hspace{0.2cm} \forall x_j$, are well defined for all $\lambda$.
% \eqref{NScond}  (see Sec. \ref{Multi:NLresourcesSubsets}) Remember that the choice of nonlocal resource inside a group of parties leads to different definitions of genuine multipartite nonlocality (see Sec. \ref{Multi:NLresourcesSubsets}). 
 %These are avoided by imposing the no-signalling principle to be satisfied by the .

%in \cite{Gallego:PRL:070401}  it was shown that 

%\added{WHY DON'T WE INCLUDE HERE A BRIEF COMMENT ON THE PROBLEM WITH THE ORIGINAL DEFINITION THAT IS NOT SUFFERED BY THE NS CASE?}

In analogy, a tripartite system is said to be in a genuine tripartite entangled pure state if it can not be decomposed as $\ket{\psi_{123}}=\ket{\phi_{ij}}\ket{\phi_k}$, where $ijk$ is any combination of the parties. One can easily verify that local measurements on biseparable states always lead to a hybrid joint distribution \eqref{GenuineNonlocal}.\\

Before introducing our inequalities witnessing genuine tripartite nonlocal correlations, recall that a Bell inequality is described by a bounded linear combination of probability terms
\begin{equation}\label{Bellineq}
\sum_{\vec{a},\vec{x}}c_{\vec{a},\vec{x}}P(\vec{a}|\vec{x})\leq \B
\end{equation}
in a experiment with $n$ parties, where the number of observables $\vec{x}=(x_1,x_2,\ldots,x_n)$ and respective outcomes $\vec{a}=(a_1,a_2,\ldots,a_n)$ is fixed. The coefficients $c_{\vec{a},\vec{x}}$ are real numbers and $\B$ is the maximum attained by local or hybrid distributions, according to the problem. 
In the bipartite scenario, where each party has two choices of two-outcome measurements, i.e. $x_i,a_i\in\{0,1\}$ for $i=1,2$, the violation of the CHSH inequality is both necessary and sufficient for $P(a_1a_2|x_1x_2)$ to be nonlocal~\cite{CHSH,FroissartSingleCHSH,FineSingleCHSH}. This is also the inequality used to show that all pure bipartite entangled states are nonlocal~\cite{GisinThrm}. Here we use a variant of the CHSH inequality, 
\begin{equation}\label{IA1A2}
I^{A_1A_2}= P(00|00)-P(01|01)-P(10|10)-P(00|11) \leq 0
\end{equation}
that, for no-signalling distributions, is equivalent to the standard expression
\begin{equation}\label{StandardCHSH}
\textrm{CHSH} = \langle A_0B_0 \rangle + \langle A_1B_0 \rangle + \langle A_0B_1 \rangle - \langle A_1B_1 \rangle \leq 2
\end{equation}
 where $\langle A_xB_y \rangle = \sum\limits_{ab} P(a=b|xy)-P(a\neq b|xy)$.

\section{Bell inequalities for genuine tripartite nonlocality}
We start by exemplifying our method through the construction of two Bell inequalities witnessing genuine tripartite nonlocal correlations. In both cases, we use the CHSH inequality \eqref{IA1A2} as the seed. The main idea is to make enough pairs of parties to play the nonlocal game defined by the seed, such the inequality can only be violated by genuine tripartite correlations. 
The first inequality is symmetrical under permutation of the three parties , 
\begin{equation}\begin{split}\label{IA1A2A3sym}
I^{A_1A_2A_3}_{\textrm{sym}} = I^{A_1A_2}_{0|0}+I^{A_1A_3}_{0|0}+I^{A_2A_3}_{0|0}-P(000|000) \leq 0
\end{split}\end{equation}
and the second inequality is symmetrical under the permutation of parties $A_2$ and $A_3$,
\begin{equation}\begin{split}\label{IA1A2A3}
I^{A_1A_2A_3}_{\textrm{\ding{192}}} = I^{A_1A_2}_{0|0}+I^{A_1A_3}_{0|0}-P(000|000) \leq 0
\end{split}\end{equation}
where the term 
\begin{equation}\begin{split}\label{I3Liftee}
I^{A_1A_2}_{0|0} \equiv P(000|000)-P(010|010)-P(100|100)\\ -P(000|110)
\end{split}\end{equation}
represents the \textit{lifting}~\cite{Pironio:Lifting} of the seed $I^{A_1A_2}$ \eqref{IA1A2} to the tripartite scenario, by setting observer $A_3$ to measurement $x_3=0$ and outcome $a_3=0$ (and similarly for the terms $I^{A_1A_3}_{0|0} $ and $I^{A_2A_3}_{0|0}$). Intuitively, in $I^{A_1A_2A_3}_{\textrm{sym}}$ every pair of parties plays a (lifted) CHSH game while in $I^{A_1A_2A_3}_{\textrm{\ding{192}}}$  party $A_1$ acquires a central role by playing a CHSH game with every remaining party. Note that the local bound of the lifted inequalities remains the same as the local bound for the seed,  $I^{A_iA_j}_{0|0}\leq0$. See Appendix \ref{APP1:liftings} for the proof of this property and more details on lifted Bell inequalities.

\begin{theorem}\label{3partyisGMNL}
The Bell inequalities $I^{A_1A_2A_3}_{\textrm{sym}}$ \eqref{IA1A2A3sym} and $I^{A_1A_2A_3}_{\textrm{\ding{192}}}$ \eqref{IA1A2A3} witness genuine tripartite nonlocality.
\end{theorem}
\begin{proof} We want to show that any hybrid distribution satisfies $I^{A_1A_2A_3}_{\textrm{sym}} \leq 0$ and $I^{A_1A_2A_3}_{\textrm{\ding{192}}} \leq 0$. First, observe that due to the convexity of hybrid distributions \eqref{GenuineNonlocal}, it is sufficient to perform the proof for the extremal distributions $P_{A_iA_j}(a_ia_j|x_ix_j)P_{A_k}(a_k|x_k)$.

The second basic element of our proof is that
\begin{equation}\label{GGG:genralLiftings0}
I^{A_jA_k}_{0|0} \big(  P_{A_iA_j}(a_ia_j|x_ix_j)P_{A_k}(a_k|x_k)  \big) \leq 0
\end{equation}
for any triplet $i,j,k \in \{1,2,3\}$ with $i \neq j \neq k \neq i$. This comes from the fact that the (lifted) inequality $I^{A_jA_k}_{0|0}$ can only be violated if parties $A_j$ and $A_k$ are non-classically correlated, which is not the case when the correlations allow for a decomposition of the form $P_{A_jA_k}(a_ja_k|x_jx_k)P_{A_i}(a_i|x_i)$. (A proof of this property can be found in Appendix \ref{APP1:liftings}.)

After this observation, we know that for every extremal hybrid distribution, the only potentially positive term in both our inequalities is 
$I^{A_iA_j}_{0|0} \big(  P_{A_iA_j}(a_ia_j|x_ix_j)P_{A_k}(a_k|x_k)\big)$, therefore:
\begin{equation}\begin{split}
I^{A_1A_2A_3}_{\textrm{sym}} \leq I^{A_iA_j}_{0|0} - P(000|000)\\
I^{A_1A_2A_3}_{\textrm{\ding{192}}} \leq I^{A_iA_j}_{0|0} - P(000|000)
\end{split}\,.\end{equation}
The last element of our proof is then the fact that
\begin{equation}\begin{split}\label{IbarA1A2A3}
\bar{I}^{A_iA_j}_{0|0} \equiv & I^{A_iA_j}_{0|0} - P(000|000)=\\ = -P(010|010)-&P(100|100)-P(000|110) \leq 0
\end{split}\end{equation}
for any probability distribution, which will cancel the potentially positive term. \end{proof}

Notice that the idea behind our construction can be extended to build a three-parameter $\mu_{12},\mu_{13},\mu_{23}$ family of inequalities that also witness genuine tripartite nonlocality
\begin{equation}\begin{split}\label{IA1A2A3FULL}
I^{A_1A_2A_3}_{\mu} = \mu_{12} I^{A_1A_2}_{0|0}+\mu_{13} I^{A_1A_3}_{0|0}\\+\mu_{23} I^{A_2A_3}_{0|0}-P(000|000) \leq 0
\end{split}\end{equation}
for $\mu_{ij} \in [0,1]$ and $\mu_{12}+\mu_{13}+\mu_{23} > 1$ (if this last condition is not met, the inequality is trivial: $I^{A_1A_2A_3}_{\mu} \leq 0$ for any probability distribution). Indeed, following the arguments of the proof of Theorem \ref{3partyisGMNL}, one can verify that \eqref{IA1A2A3FULL} holds for any hybrid distribution \eqref{GenuineNonlocal}:
\begin{equation}\begin{split}\label{IA1A2A3FULLifull}
I^{A_1A_2A_3}_{\mu}\big(  P_{A_iA_j}(a_ia_j|x_ix_j)P_{A_k}(a_k|x_k)  \big)\leq &\\ \mu_{ij}I^{A_iA_j}_{0|0}-P(000|000) \leq \bar{I}^{A_iA_j}_{0|0} \leq& 0
\end{split}\end{equation}
%REM2: Note that the combination $\mu=\nu=0$ is not a Bell inequality any more.

%$I^{A_1A_2A_3}_{\mu,\nu} \big(  P_{A_iA_j}(a_ia_j|x_ix_j)P_{A_k}(a_k|x_k)  \big)$ on any hybrid distribution $P_{A_iA_j}(a_ia_j|x_ix_j)P_{A_k}(a_k|x_k)$.\\
\vspace{0.3cm}
It is interesting to observe that the local strategy where every party always obtains outcome $a_i=1$ for any measurement $x_i$ saturates both inequalities $I^{A_1A_2A_3}_{\textrm{\ding{192}}} = 0$ and $I^{A_1A_2A_3}_{\textrm{sym}} = 0$, thus as well $I^{A_1A_2A_3}_{\mu} = 0$. This implies that the local and hybrid bounds of our inequalities coincide. It also implies that our inequalities are tangent, both to the set of local correlations \eqref{FullyLocal} and hybrid ones \eqref{GenuineNonlocal}.\\

\section{Detection of genuine tripartite nonlocality in pure states}
Bell inequality $I^{A_1A_2A_3}_{\textrm{sym}}$ seems particularly fit for the detection of genuine tripartite nonlocality of pure states.  Indeed, it belongs to class 6 of \cite{Bancal:PRA:014102}, where strong numerical evidence was provided indicating that all three-qubit systems in a GME pure state could generate correlations violating it. This result hints that $I^{A_1A_2A_3}_{\textrm{sym}}$ is a good candidate for an analytical proof of equivalence between GME and GMNL for tripartite pure states. Later, we will generalise this inequality for $n$ parties, prove analytically that it detects GMNL in a large class of GME pure states and provide numerical evidence that all GME  pure states of four qubits are GMNL. 

We now focus on inequality $I^{A_1A_2A_3}_{\textrm{\ding{192}}} \leq 0$ and show that it is useful for the detection of genuine tripartite nonlocality of pure states. In \cite{Acin2000}, it was shown that all systems of three qubits in a pure state could be written as
\begin{equation}\begin{split}\label{3qubitsMAIN}
\ket{\Psi_3} = h_0 \ket{000} + h_1 e^{i \phi} \ket{100} + h_2 \ket{101} + h_3 \ket{110} + h_4 \ket{111}
\end{split}\end{equation}
where $h_i \in \mathbb{R}_{+}$, $\sum\limits_i h_i^2 = 1$ and $\phi \in [0,\pi]$.

\begin{theorem}\label{GG:THRM2} For all tripartite pure states \eqref{3qubitsMAIN} that are GME and  symmetrical under the permutation of any two parties, say $A_2$ and $A_3$ ($h_0,h_4>0$ and $h_2=h_3$), one can find local measurements on them such that the generated correlations violate inequality $I^{A_1A_2A_3}_{\textrm{\ding{192}}} \leq 0$ \eqref{IA1A2A3}, hence generating GMNL correlations.
\end{theorem}
%REM3: 

\begin{proof}
A complete proof of Theorem \ref{GG:THRM2} can be found in Appendix \ref{APPGG:TheoremN3}. The main line of it goes as follows. We will start by choosing parties $A_2$ and $A_3$ to perform the same (projective) measurements $\bra{m_{a_2|x_2}} = \bra{m_{a_3|x_3}} \hspace{0.2cm}, \forall a_2=a_3$ and $x_2=x_3$. This, together with the $A_2 \leftrightarrow A_3$ invariance of the state, implies that the observed correlations $P(a_1a_2a_3|x_1x_2x_3)$ are also symmetrical with respect to the permutation of $A_2$ and $A_3$. Consequently, for these correlations, we have $I^{A_1A_2}_{0|0}=I^{A_1A_3}_{0|0}$ \eqref{I3Liftee} and $I^{A_1A_2A_3}_{\textrm{\ding{192}}}$ can be simplified to 
\begin{equation}\begin{split}
I^{A_1A_2A_3}_{\textrm{\ding{192}}}  = 2 I^{A_1A_2}_{0|0}-P(000|000) =\\=P(000|000)-\\-2P(010|010)-2P(100|100)-2P(000|110)\leq 0
\end{split}\end{equation}
We now show that we can always find appropriate measurements such that we obtain a particular violation of the previous inequality
\begin{equation}\begin{split}\label{HARDYMEASMAIN}
P(000|000) & > 0\\
P(010|010) &= P(100|100) = P(000|110) = 0\,.
\end{split}\end{equation}
These conditions correspond to an Hardy paradox \cite{Hardy} on parties $A_1$ and $A_2$. 
Consider the post-measurement state $\ket{\psi_{0|0}^{A_1A_2}}$, which is the state prepared by party $A_3$ after making the measurement $\bra{m_{a_3=0|x_3=0}}$ on $\ket{\Psi_3}$:
\begin{equation}
\ket{\psi_{0|0}^{A_1A_2}} \propto \id_{A_1} \otimes \id_{A_2} \otimes \bra{m_{a_3=0|x_3=0}} \ket{\Psi_3}\,.
\end{equation}
Since $\ket{\Psi_3}$ is GME by assumption, we can tune the measurement $\bra{m_{a_3=0|x_3=0}}$ such that the  prepared state $\ket{\psi_{0|0}^{A_1A_2}}$ is non-maximally entangled \cite{PopescuGeneric}. After Hardy's construction \cite{Hardy}, we know that for a pure non-maximally entangled state $\ket{\psi_{0|0}^{A_1A_2}}$ we can always find a one-parameter family of measurements on $\bra{m_{a_1|x_1}}$ and $\bra{m_{a_2|x_2}}$ leading to an Hardy paradox \eqref{HARDYMEASMAIN}. This means that we can choose freely the first measurement, say $\bra{m_{a_2=0|x_2=0}}$, and always find three other measurements $\bra{m_{a_2=0|x_2=1}},\bra{m_{a_1=0|x_1=0}},\bra{m_{a_1=0|x_1=1}}$ such that \eqref{HARDYMEASMAIN} is satisfied. %To be precise, for any pure (non-maximally) entangled state $\ket{\psi_{0|0}^{A_1A_2}}$ the measurements can not be chosen completely freely: there is one point in the full space of measurements that is not allowed. Being only a point in the full space, this does not affect the proof. 
Therefore, we are able to choose $\bra{m_{a_2=0|x_2=0}}=\bra{m_{a_3=0|x_3=0}}$ in order to be compatible with the condition of preparing a state $\ket{\psi_{0|0}^{A_1A_3}}$ $(=\ket{\psi_{0|0}^{A_1A_2}})$ that is non-maximally entangled.
More details can be found in Appendices \ref{APPGG:TheoremN3} and  \ref{GGAPP:HardyMeasurements} .
%
%
%This symmetry also implies that the two prepared states $\ket{\psi_{0|0}^{A_1A_2}} = \ket{\psi_{0|0}^{A_1A_3}}$ are the same, where the state $\ket{\psi_{0|0}^{A_1A_2}}$ between $A_1$ and $A_2$ is the state prepared by party $A_3$ by making the measurement $\bra{m_{a_3=0|x_3=0}}$ on $\ket{\Psi_3}$ \eqref{3qubitsMAIN},
%\begin{equation}
%\ket{\psi_{0|0}^{A_1A_2}} \propto \id_{A_1} \otimes \id_{A_2} \otimes \bra{m_{a_3=0|x_3=0}} \ket{\Psi_3}\,,
%\end{equation}
%and similarly for $\ket{\psi_{0|0}^{A_1A_3}}$. Since $\ket{\Psi_3}$ is GME by assumption, we can always tune the measurements  $\bra{m_{a_2=0|x_2=0}} = \bra{m_{a_3=0|x_3=0}}$ such that the  prepared states $\ket{\psi_{0|0}^{A_2A_3}} = \ket{\psi_{0|0}^{A_1A_2}}$ are non-maximally entangled \cite{PopescuGeneric}. 
%
%
%

\end{proof}

\section{The general multipartite scenario} 
We proceed now to exposing our results in the general multipartite scenario. We consider any number $n>2$ of observers making dichotomic choices of local measurements $x_i \in \{0,1\}$ on their share of a joint quantum system and obtaining outcomes $a_i \in \{0,1\}$, generating a distribution $P(a_1a_2...a_n|x_1x_2...x_n) \equiv P(\vec{a}|\vec{x})$. 

The definition of genuine multipartite nonlocality for any number of parties is more intricate than the tripartite case, but follows basically the same idea as \eqref{GenuineNonlocal}. A distribution $P(\vec{a}|\vec{x})$ is said to be \textit{biseparable} if
\begin{equation}\begin{split}\label{biseparable}
P_{\textrm{2-sep}}(\vec{a}|\vec{x}) = \sum\limits_{g}\sum\limits_{\lambda_g} q_g(\lambda_g) P(\vec{a}_g|\vec{x}_g,\lambda_g)P(\vec{a}_{\bar{g}}|\vec{x}_{\bar{g}},\lambda_{g})
\end{split}\end{equation}
where $\sum\limits_{g}\sum\limits_{\lambda_g} q_g(\lambda_g) = 1$, $q_g(\lambda_g) \geq 0$ and $g$ is a group consisting of a particular subset of the $n$ observers and $\bar{g}$ its complement. We label the string of measurement choices (resp. outcomes) of the observers belonging to the group $g$  as $\vec{x}_g$ ($\vec{a}_g$). For example, for $n=3$ there are only three possible inequivalent ways of making two groups: $(g_1=A_1A_2,\bar{g}_1=A_3$), $(g_2=A_1A_3, \bar{g}_2=A_2$) and $(g_3=A_2A_3, \bar{g}_3=A_1$), leading to a decomposition of the form \eqref{GenuineNonlocal}. Distributions that can not be written according to the decomposition \eqref{biseparable} are \emph{genuine multipartite nonlocal}. This corresponds to the strongest type of multipartite nonlocality, in the sense that all the parties of the system are engaged in a nonclassical correlation. We will see later that we can define intermediate types of multipartite nonlocal correlations, where one allows for more than two groups of parties.% are allowed to share non-classical correlations, but classical correlations only are allowed between the groups.

Again, local measurements on pure biseparable states $\ket{\psi_{1\ldots n}}=\ket{\phi_g}\ket{\phi_{\bar{g}}}$ for some splitting $g/\bar{g}$ of the particles, always lead to biseparable joint distributions \eqref{biseparable}. Genuine multipartite entanglement is necessary to generate genuine multipartite nonlocal correlations.\\

\section{Bell inequalities for genuine multipartite nonlocality}
The generalisation of inequalities $I^{A_1A_2A_3}_{\textrm{sym}}$ \eqref{IA1A2A3sym} and $I^{A_1A_2A_3}_{\textrm{\ding{192}}}$ \eqref{IA1A2A3} to any number $n$ of parties gives two distinct families of Bell inequalities that can be written in a simple form:
\begin{equation}\begin{split}\label{IGeneralsym}
I^{A_1\ldots A_n}_{\textrm{sym}} =\sum\limits_{i=1}^{n-1}\sum\limits_{j>i}^{n} I^{A_iA_j}_{\vec{0}|\vec{0}} - {n-1\choose 2}  P(\vec{0}|\vec{0}) \leq 0
\end{split}\end{equation}
\begin{equation}\begin{split}\label{IGeneral}
I^{A_1\ldots A_n}_{\textrm{\ding{192}}} = \sum\limits_{j>1}^{n} I^{A_1A_j}_{\vec{0}|\vec{0}} - (n-2)  P(\vec{0}|\vec{0}) \leq 0
\end{split}\end{equation}
where ${n-1\choose 2}  = \frac{(n-1)(n-2)}{2}$ and we take the freedom of writing $\vec{0}\equiv (0,0,...,0)$, the size of the string should be obvious in the context. Similarly to \eqref{I3Liftee}, $I^{A_iA_j}_{\vec{0}|\vec{0}}$ is a lifting of inequality $I^{A_iA_j}$ \eqref{IA1A2} to $n$ observers by setting the remaining $n-2$ observers to have their measurement and outcome set to $0$:
\begin{equation}\begin{split}\label{InLiftee}
I^{A_iA_j}_{\vec{0}|\vec{0}} = P(0_i0_j\vec{0}|0_i0_j\vec{0})-P(1_i0_j\vec{0}|1_i0_j\vec{0})\\ -P(0_i1_j\vec{0}|0_i1_j\vec{0})-P(0_i0_j\vec{0}|1_i1_j\vec{0})\,.
\end{split}\end{equation}
The operational meaning of these inequalities is the following. In both cases, we use as ``seed'' the CHSH inequality \eqref {IA1A2}, which defines a nonlocal game between parties $A_i$ and $A_j$ represented by the lifted inequalities $I^{A_iA_j}_{\vec{0}|\vec{0}}$. For the symmetrical family \eqref{IGeneralsym} every pair of parties is required to play a CHSH game while for the inequalities ``centered" on $A_1$ \eqref{IGeneral}, party $A_1$ is required to play a CHSH game with every other party.

%if more than ${n-1\choose 2}$ different pairs $A_iA_j$ are be able to win (i.e. beat the local bound) then the correlations are GMNL. if more than $n-2$ different pair of parties $A_1A_j$ win the game, then the correlations are GMNL.

%IMAGE?% (n=5, sym vs centered avec des graph ou les parties jouent une Ineq liftée)\\
%For $n = 2$ and $n = 3$, our family of inequalities \eqref{IGeneralsym} coincides with inequalities \eqref{IA1A2} and \eqref{IA1A2A3} respectively. Finally, note that
%\begin{equation}\begin{split}\label{InBarLiftee}
%\bar{I}^{A_iA_j}_{\vec{0}|\vec{0}} \equiv I^{A_iA_j}_{\vec{0}|\vec{0}} - P(\vec{0}|\vec{0}) = -P(1_i0_j\vec{0}|1_i0_j\vec{0})\\ -P(0_i1_j\vec{0}|0_i1_j\vec{0})-P(0_i0_j\vec{0}|1_i1_j\vec{0}) \leq 0
%\end{split}\end{equation}
%for any valid probability distribution $P(\vec{a}|\vec{x})$.\\ 

\begin{theorem}\label{NpartyisGMNL}
The Bell inequalities $I^{A_1\ldots A_n}_{\textrm{sym}}\leq 0$ \eqref{IGeneralsym} and $I^{A_1\ldots A_n}_{\textrm{\ding{192}}}\leq 0$ \eqref{IGeneral} are witnesses of genuine multipartite nonlocality for all $n\geq3$.
\end{theorem}
\begin{proof}
Here we only provide an outline, the full proof can be found in Appendix \ref{APP2:IGeneral}. 
The idea is similar to the one for three parties. We want to show that all biseparable distributions \eqref{biseparable} for $n$ parties satisfy $I^{A_1\ldots A_n}_{\textrm{sym}} \leq 0$ and $I^{A_1\ldots A_n}_{\textrm{\ding{192}}}\leq 0$. Again, by convexity, it is enough to verify it for extremal biseparable distributions $P(\vec{a}_g|\vec{x}_g)P(\vec{a}_{\bar g}|\vec{x}_{\bar g})$ from \eqref{biseparable}.

%\begin{equation}\label{Iforbisep}
%I^{A_1\ldots A_n}_{\textrm{sym}}\big(P(\vec{a}_g|\vec{x}_g)P(\vec{a}_{\bar g}|\vec{x}_{\bar g})\big)\leq0\,.
%\end{equation}
If parties $A_i$ and $A_j$ belong to different groups, they are only classically correlated and therefore $I^{A_iA_j}_{\vec{0}|\vec{0}}\leq0$. Then, the only terms that can give a positive contribution $I^{A_iA_j}_{\vec{0}|\vec{0}} > 0$ are terms where parties $A_i$ and $A_j$ belong to the same group. Now the trick is to kill these positive contributions by subtracting enough $P(\vec0|\vec0)$ terms since, similarly to $n=3$ \eqref{IbarA1A2A3},
\begin{equation}\label{Ibar}
\bar{I}^{A_iA_j}_{\vec{0}|\vec{0}}\equiv I^{A_iA_j}_{\vec{0}|\vec{0}} - P(\vec{0}|\vec{0}) \leq0 
\end{equation}
%according to a direct extension of \eqref{IbarA1A2A3} to $n$ parties. 
for any probability distributions.

\emph{Symmetric family} $I_{\textrm{sym}}$ --- In general, if the first group $g$ consists of $m$ parties and $\bar{g}$ of $n-m$ for some $1 \leq m \leq n-1$, a total number of ${m\choose 2}+{n-m\choose 2}$ inequalites $I^{A_iA_j}_{\vec{0}|\vec{0}}$ can in principle be positive. Since ${n-1\choose 2} > {m\choose 2}+{n-m\choose 2} \hspace{0.2cm}, \forall m \geq 2$, the largest number of pairs is obtained by putting $n-1$ parties in one group, which means ${n-1\choose 2}$ potentially positive terms $I^{A_iA_j}_{\vec{0}|\vec{0}}$. Then, 
\begin{equation}\begin{split}
I^{A_1\ldots A_n}_{\textrm{sym}}\big(P(\vec{a}_g|\vec{x}_g)&P(\vec{a}_{\bar g}|\vec{x}_{\bar g})\big)\leq \\ \sum\limits_{i=1}^{n-2}\sum\limits_{j>i}^{n-1} I^{A_iA_j}_{\vec{0}|\vec{0}} - {n-1\choose 2} &P(\vec{0}|\vec{0}) =\sum\limits_{i=1}^{n-2}\sum\limits_{j>i}^{n-1} {\bar I}^{A_iA_j}_{\vec{0}|\vec{0}} \leq 0
\end{split}\end{equation}
where we used the fact that $I^{A_1\ldots A_n}_{\textrm{sym}}$ is invariant under permutations of parties to consider the specific partition $g=\{1,..,n-1\}$ and $\bar g =\{n\}$. \\\
%REM5:One can check that ${n-1\choose 2} > {m\choose 2}+{n-m\choose 2} \hspace{0.2cm} \forall m \geq 2$.

\emph{Centered family} $I_{\textrm{\ding{192}}}$ --- The proof follows the same idea as before. Using \eqref{Ibar}, any biseparable distribution \eqref{biseparable} with $m$ parties in the first group $g$ containing party $A_1$ and $n-m$ in the other group $\bar{g}$ gives
\begin{equation}\begin{split}\label{IbarNonSym}
I^{A_1\ldots A_n}_{\textrm{\ding{192}}}\big(P(\vec{a}_g|\vec{x}_g) P(\vec{a}_{\bar g}|\vec{x}_{\bar g})\big) \\ \leq \sum\limits_{j \in g} I^{A_1A_j}_{\vec{0}|\vec{0}} - (n-2)  P(\vec{0}|\vec{0}) \leq \sum\limits_{j \in g} &\bar{I}^{A_1A_j}_{\vec{0}|\vec{0}} \leq 0
\end{split}\end{equation}
since there are at most $n-2$ parties together with party $A_1$ in the first group $g$.% since the distribution $P(\vec{a}_g|\vec{x}_g)P(\vec{a}_{\bar g}|\vec{x}_{\bar g})$ is biseparable (for any choice of groups $g$ and $\bar{g}$).

\end{proof}\vspace{0.5cm}
%. In the scenario with two observers, it is known that one can find local measurements on all two-particles systems in a pure entangled state such that the generated distribution is nonlocal as witnessed by a violation of \eqref{IA1A2}. For three observers, the inequality was shown numerically to witness GMNL distributions from many three-qubit systems in a genuine multipartite entangled state, showing evidence that all such states can generate GMNL distributions as witnessed by this inequality alone. 

%One can understand a violation of the families \eqref{IGeneralsym} and \eqref{IGeneral} of inequalities in the following way: GMNL correlations are the only ones for which it is potentially possible to violate a lifted inequality $I^{A_iA_j}_{\vec{0}|\vec{0}} > 0$ between \textit{all} pairs of parties, as these correlations are the ones where all the parties share nonlocal resources. Biseparable correlations \eqref{biseparable} are limited in this sense, as many parties are only classically correlated to the parties that are in a different group and thus numerous lifted inequalities $I^{A_iA_j}_{\vec{0}|\vec{0}}$ cannot be violated. We give an illustration of that argument in Fig. \ref{Fig:GG_n5}.\\

\begin{figure}[h!]\begin{center}
  \scalebox{0.7}{\includegraphics{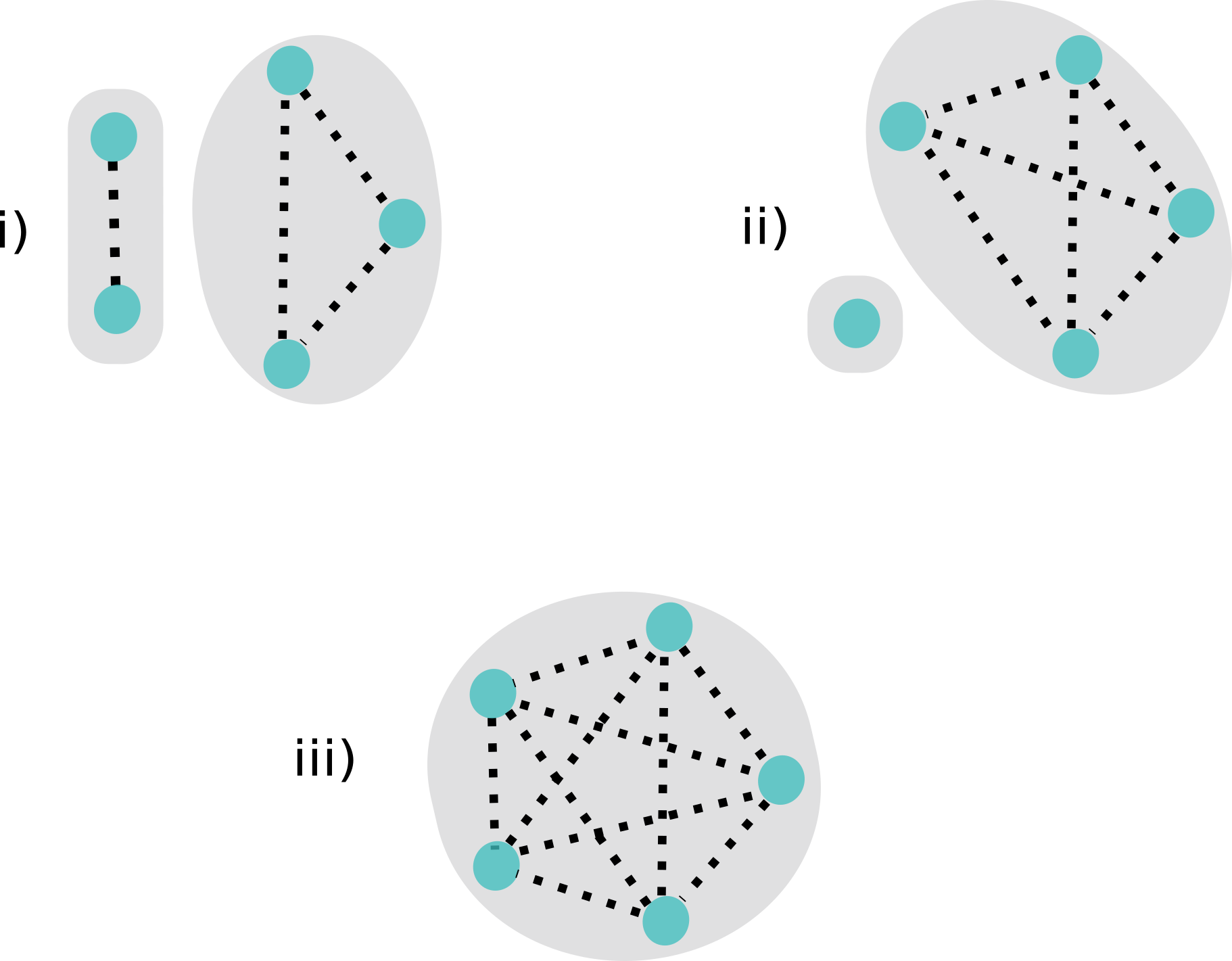}}
   \caption{\label{Fig:GG_n5} An abstract representation of five parties (blue circles) arranged into groups (gray areas). Two parties $i$ and $j$ inside the same group can potentially violate a lifted inequality $I^{A_iA_j}_{\vec{0}|\vec{0}}$ (as represented by a dashed line between them). $i)$ Two groups of parties $|g|=2;|\bar{g}|=3$, giving a distribution of the form $P(\vec{a}|\vec{x})=P(a_1a_2|x_1x_2)P(a_3a_4a_5|x_3x_4x_5)$ and a maximum number of ${2\choose 2}+{3\choose 2} = 4$ violated inequalities $I^{A_iA_j}_{\vec{0}|\vec{0}}>0$. $ii)$ Two groups of parties $|g|=1;|\bar{g}|=4$, for ${4 \choose 2} = 6$ potentially violated inequalities $I^{A_iA_j}_{\vec{0}|\vec{0}}>0$. $iii)$ GMNL: all parties are in the same group and thus ${5 \choose 2} = 10$ inequalities can be violated. Only $iii)$ can violate $I^{A_1\ldots A_5}_{\textrm{sym}} = \sum\limits_{i=1}^4 \sum\limits_{j>i}^5 I^{A_iA_j}_{\vec{0}|\vec{0}} - 6P(\vec{0}|\vec{0})$ since $I^{A_iA_j}_{\vec{0}|\vec{0}} - P(\vec{0}|\vec{0}) \leq 0$.}
  \end{center}\end{figure}

One can understand a violation of the families of inequalities \eqref{IGeneralsym} and \eqref{IGeneral} in the following way: $a)$ any (extremal) distribution that violates $I^{A_1\ldots A_n}_{\textrm{sym}}$ \eqref{IGeneralsym} needs to be capable of violating more than ${n-1\choose 2}$ lifted CHSH inequalities $I^{A_iA_j}_{\vec{0}|\vec{0}}$ between different observers $A_i$ and $A_j$; and $b)$ any distribution that violates $I^{A_1\ldots A_n}_{\textrm{\ding{192}}}$ \eqref{IGeneral} violates more than $n-2$ lifted CHSH inequalities $I^{A_1A_j}_{\vec{0}|\vec{0}}$ between $A_1$ and different observers $A_j$. Only GMNL correlations, where all pairs of parties are nonlocally correlated, are able to do this.% \deleted{Note that correlations that are not GMNL and not extremal might in general be able to violate lifted CHSH inequalities between all pairs of parties. However, when decomposed into extremal (non GMNL) correlations as in \eqref{biseparable}, each extremal contribution is unable to violate enough lifted CHSH inequalities, contrary to correlations that are GMNL.}\added{I'M NOT SURE IF I UNDERSTAND THIS. DO YOU MIND REMOVING IT AND WE DISCUSS IT LATER? OR YOU CAN LEAVE IT AND WE CAN DISCUSS IT AFTERWARDS. I FIND IT CONFUSING}

%IMAGE?\\

More insight on the rich structure of the symmetrical family of inequalities \eqref{IGeneralsym} is given by noticing that they can also be written in a recursive form for $n\geq3$
\begin{equation}\begin{split}\label{IGeneralsym2}
I^{A_1A_2...A_n}_{\textrm{sym}} = \sum\limits_{i = 1}^{n} I^{\textrm{all}\setminus A_i}_{0|0} - (n-2)P(\vec{0}|\vec{0}) \leq 0
\end{split}\end{equation}
where $I^{\textrm{all}\setminus A_i}_{0|0}$ is the symmetrical inequality for $n-1$ observers lifted to $n$ of them with observer $A_i$'s input and outcome set to $0$. If $n=3$ for example, $I^{\textrm{all}\setminus A_i}_{0|0}$ corresponds to the CHSH inequality lifted to 3 parties \eqref{I3Liftee}.
The proof of the equivalence between the direct expression \eqref{IGeneralsym} and the recursive one \eqref{IGeneralsym2} can be found in Appendix~\ref{APP:Alternative}.

In other words, operationally a violation of the symmetrical family $I^{A_1A_2...A_n}_{\textrm{sym}}$ can also be understood as a violation of more than $n-2$ inequalities $I^{\textrm{all}\setminus A_i}_{0|0}$ between $n-1$ parties lifted to $n$ parties -- instead of ${n-1\choose 2}$ bipartite ones $I^{A_iA_j}_{\vec{0}|\vec{0}}$ lifted to $n$ parties. Since this argument can be used recursively, one concludes that GMNL correlations violating our inequalities violate numerous inequalities between subset of $m$ parties lifted to $n$ parties, for all $m$.% \deleted{Again, non extremal biseparable correlations might also violate many inequalities between subsets of parties. Nevertheless, each extremal (non GMNL) contribution in their decomposition cannot.}\added{SAME HERE}

Observe that, similar to the tripartite case, the generalised families $I^{A_1\ldots A_n}_{\textrm{sym}}$ \eqref{IGeneralsym} and $I^{A_1\ldots A_n}_{\textrm{\ding{192}}}$ \eqref{IGeneral} are also saturated by local distributions. The local strategy is the same as for $n=3$: every party $A_i$ outputs $a_i=1$ for all measurements $x_i$. It follows that the local and biseparable bounds of our families of inequalities coincide for all $n$. Again, it also implies that our inequalities are tangent to the set of local and biseparable correlations.

\section{Detection of genuine multipartite nonlocality in pure states}

Let us now analyse how the symmetric family of Bell inequalities $I^{A_1\ldots A_n}_{\textrm{sym}}$ sheds light on the relation between GME and GMNL of pure states.  The goal is to understand  whether it is possible to find local measurements on \emph{any} pure GME state that generate GMNL correlations. Our Bell inequalities seem fit to prove this result since for any pure GME state there exist local projections on \emph{any} $n-2$ parties that leave the remaining two in a pure entangled state \cite{PopescuGeneric}, which can in turn be used to violate the CHSH inequality~\cite{GisinThrm}.  The main difficulty in proving the result in full generality is to find local measurements that simultaneously perform the desired projections but are also fit to violate the CHSH terms. For $n=2$ our two families of inequalities coincide with the CHSH inequality, which was used to prove the equivalence between nonlocality and pure state entanglement~\cite{GisinThrm}. For $n=3$, there is numerical evidence that this holds for GME three-qubits pure states \cite{Bancal:PRA:014102} using the symmetrical family $I^{A_1A_2A_3}_{\textrm{sym}}$ \eqref{IA1A2A3}.\\

We consider the generalisation of these results to the scenario with $n=4$ parties, where we obtained numerical evidence that all four-qubit systems in a pure GME state generate distributions violating the Bell inequality $I^{A_1\ldots A_4}_{\textrm{sym}}$ \eqref{IGeneralsym}. For this, we have randomly drawn four qubit states and numerically searched for local measurements leading to a violation of our inequality. Note that the set of separable states is of volume zero in the state space.\\

%\textbf{To do so, we have, repeatedly, drawn 16 real numbers at random forming the eight coefficients $c_{k_1,k_2}^{k_3,k_4} \in \mathbb{C}, k_i=0,1$ of a pure state $\ket{\psi} = \sum\limits_{k_1k_2k_3k_4} c_{k_1,k_2}^{k_3,k_4} \ket{k_1k_2k_3k_4}$, then normalised the state and numerically searched for projective measurements leading to a violation of inequality \eqref{IGeneralsym}. Doing so, we have obtained a violation of our inequality for XXXX pure states without finding a counter-example. (Note that the set of biseparable states is of volume zero in the full space of states.)}\\

We now proceed to show analytically that a large class of pure GME states of the GHZ family~\cite{greenberger1989going} can generate GMNL correlations, for all number of parties $n\geq3$, as detected by the symmetrical family of inequalities $I_{\textrm{sym}}^{A_1\ldots A_n}$ \eqref{IGeneralsym} . 

\begin{theorem}\label{GG:THRM1} All pure GME states of the form
\begin{equation}\begin{split}\label{GG:greenberger1989going}
\ket{GHZ^n}_{\theta} = \cos\theta\ket{0}^{\otimes n} - \sin\theta\ket{1}^{\otimes n}
\end{split}\end{equation}
with $\theta \in ]0,\frac{\pi}{4}[$ violate the Bell inequality $I_{\textrm{sym}}^{A_1\ldots A_n}$ \eqref{IGeneralsym} for all $n\geq 3$.
All parties $A_i$ make the same projective measurements, $\bra{m_{a_i|x_i}}= \bra{m_{a|x}}$,  defined by
\begin{align}
\begin{split}\label{Measurements}
\bra{m_{0|x}} = \cos\alpha_x\bra{0}+\sin\alpha_x\bra{1}\\
\bra{m_{1|x}} = \sin\alpha_x\bra{0}-\cos\alpha_x\bra{1}
\end{split}
\end{align}
where
\begin{align}
\begin{split}\label{angles}
\alpha_0 &= \arctan(\tan^{-\frac{3}{3n-4}}(\theta))\\
\alpha_1 &= -\arctan(\tan^{-\frac{1}{3n-4}}(\theta))\,.
\end{split}
\end{align}
In other words, all states of the form \eqref{GG:greenberger1989going} that are GME are GMNL.
\end{theorem}

\begin{proof} A detailed and constructive proof of this theorem can be found in Appendix \ref{APP3:GME-GMNL}. The key point is to impose the local measurements to be the same for every party, which makes the joint outcome distribution invariant under permutations of the parties (since the state $\ket{GHZ^n}_{\theta}$ \eqref{GG:greenberger1989going} also has this invariance). This symmetry simplifies the problem and allowed us to find an analytical solution.
\end{proof} 

Interestingly, the only GME pure state of this family for which our construction fails is the maximally entangled state ($\theta = \pi/4$), which is already known to generate GMNL for any number of observers \cite{BancalQuantifying}. We have however found, numerically, several sets of measurements on this state that lead to distributions  violating our inequality, but the amount of symmetries is reduced. Interestingly, Theorem \ref{GG:THRM1} implies that even states that are almost separable ($\theta \rightarrow 0$) can be used to generate GMNL correlations for any number of observers.\\

%Note that a larger class of symmetrical states -- including the ones of \eqref{GG:greenberger1989going} -- was already shown to be GMNL for all $n$ in \cite{Chinese1,Chinese2} based on the violation of Hardy-like paradoxes witnessing GMNL.\\

It is important to observe that we already knew from \cite{Chinese1} that all $n$-qubit systems in a GME symmetric pure state are GMNL, which is a more general result than Theorem \ref{GG:THRM1}. In particular, for three-qubit systems, the problem was completely solved: all three-qubit systems in a pure GME state can exhibit GMNL \cite{Chinese2}. These results rely however on the violation of two families of Hardy-like paradoxes witnessing GMNL, making it untestable in an experiment  (where even the smallest imperfections lead to values $P(ab|xy)= \epsilon > 0$). Indeed, the correlations obtained in~\cite{Yu3,Chinese1} are nonlocal, because they violate a Hardy paradox, and therefore should also violate a Bell inequality in a robust way. However, the form of this inequality is unknown. Although clearly not as general, our results are testable and might lead the way to a complete generalisation of Gisin's theorem for $n$ parties.

\section{Constructing genuine multipartite Bell inequalities from different seeds}\label{OtherSeed}
%I_theta, Ineq 5/6? de NS2/1

The construction was so far done using as ``seed" the CHSH inequality $I^{A_1A_2}$ \eqref{IA1A2} to  build new families of inequalities. We now show the versatility of our technique by using different inequalities as seed.  In general, any inequality $S^{A_1\ldots A_m}$ for $m<n$ parties that can be written as
\begin{equation}\begin{split}\label{GG:UsefulIneqs}
S^{A_1\ldots A_m}\big(P(\vec{a}|\vec{x})\big) &=\\= P(\vec0|\vec0) &- \sum\limits_{\substack{\vec{a},\vec{x} \neq
                \vec0,\vec0}} \beta_{\vec a}^{\vec x} P(\vec a|\vec x) \leq \mathcal{B}_{2-sep}= 0\end{split}
\end{equation}
such that $\beta_{\vec a}^{\vec x} \geq 0 \hspace{0.2cm}$,$ \forall {\vec a,\vec x \neq \vec0,\vec0}$, and with biseparable bound $\mathcal{B}_{2-sep}=0$, is a valid seed to build a Bell inequality for $n$ parties. To see that, note that the key ingredient in our proofs is, again, that
\begin{equation}\label{GG:UsefulIneqs2}
\bar{S}^{A_1\ldots A_m} \equiv S^{A_1\ldots A_m} - P(\vec0|\vec0) =  - \sum\limits_{\substack{\vec a,\vec x \neq \vec 0,\vec 0}} \beta_{\vec a}^{\vec x} P(\vec a|\vec x) \leq 0
\end{equation}
for any probability distribution. (This implies that the lifting of \eqref{GG:UsefulIneqs2} to more parties is also always negative $\bar{S}^{A_1\ldots A_n}_{\vec{0}|\vec{0}} \leq 0$, which we used frequently in our proofs). 
Although the condition for a Bell inequality to be used as a seed is fairly simple, we do not know of a systematic way to find out which inequalities can be written in the form \eqref{GG:UsefulIneqs}.
We will now illustrate our construction with two different seeds.\\

\textbf{The tilted CHSH inequality.--} As a first example, we use the ``tilted CHSH" inequality  \cite{Acin2012} as the new seed. This inequality is a variation of the CHSH inequality with two free parameters, used for randomness certification~\cite{Acin2012} and self-testing of partial entangled states \cite{bamps2015sum}. By setting one of the parameters to 1, the tilted CHSH inequality can be written in the form \eqref{GG:UsefulIneqs}
\begin{equation}\begin{split}\label{Ibeta}
I^{A_1A_2}_{\beta} = P(00|00)-P(01|01)-P(10|10)\\-P(00|11)-\frac{\beta}{2}P_{A_1}(1|0) \leq 0
\end{split}\end{equation}
where $\beta \geq 0$ and $P(a_1|x_1) = \sum\limits_{a_2} P(a_1a_2|x_1x_2) \hspace{0.2cm},\forall x_2$, is the marginal distribution of party $A_1$. Clearly, the inequality satisfies condition \eqref{GG:UsefulIneqs2}.
%As said, similarly to $\bar{I}^{A_iA_j}_{\vec{0}|\vec{0}} \leq 0$ \eqref{Ibar}
%\begin{equation}\begin{split}\label{IbarBeta}
%\bar{I}^{A_iA_j}_{\beta,\vec{0}|\vec{0}} \equiv I^{A_1A_2}_{\beta,\vec{0}|\vec{0}} - P(\vec{0}|\vec{0}) = -P(\vec{0}1_j|\vec{0}1_j)-P(1_i\vec{0}|1_i\vec{0})\\ -P(\vec{0}|1_i1_j\vec{0}) - \sum\limits_{a_j} \frac{\beta}{2}P(1_ia_j\vec{0}|\vec{0}) \leq 0
%\end{split}\end{equation}
%where $\bar{I}^{A_1A_2}_{\beta,\vec{0}|\vec{0}}$ is a lifting of inequality $I^{A_1A_2}_{\beta}$ \eqref{Ibeta} to $n$ parties, in analogy to the lifting $I^{A_1A_2}_{\vec{0}|\vec{0}}$ to $n$ parties of the CHSH inequality $I^{A_1A_2}$. The term $\sum\limits_{a_j} \frac{\beta}{2}P(1_ia_j\vec{0}|\vec{0})$ is the marginal $P_{A_i}(1|0) = \sum\limits_{a_j} P_{A_iA_j}(1a_2|0x_2) \hspace{0.2cm} \forall x_2$ of party $A_i$ lifted to $n$ parties. \\
Starting from the new seed $I^{A_1A_2}_{\beta}$ \eqref{Ibeta}, we construct two new families of GMNL Bell inequalities.

\begin{theorem}\label{GG:NewSeedTHRM}
The families of inequalities
\begin{equation}\begin{split}\label{IGeneralsymBeta}
I^{A_1\ldots A_n}_{\beta,\textrm{sym}} =\sum\limits_{i=1}^{n-1}\sum\limits_{j>i}^{n} I^{A_iA_j}_{\beta,\vec{0}|\vec{0}} - {n-1\choose 2}  P(\vec{0}|\vec{0}) \leq 0
\end{split}\end{equation}
\begin{equation}\begin{split}\label{IGeneralBeta}
I^{A_1\ldots A_n}_{\beta,\textrm{\ding{192}}} = \sum\limits_{j>1}^{n} I^{A_1A_j}_{\beta,\vec{0}|\vec{0}} - (n-2)  P(\vec{0}|\vec{0}) \leq 0
\end{split}\end{equation}
witness GMNL for any number $n\geq3$ of parties.
\end{theorem}

\begin{proof}
The proof that these two families of inequalities indeed witness GMNL for all $n$ is exactly the same as for the families \eqref{IGeneralsym} and \eqref{IGeneral}, but now using the seed $I_\beta^{A_1A_2}$ and property \eqref{GG:UsefulIneqs2}.
\end{proof}

\textbf{A tripartite inequality as seed.--} As a second example, we illustrate how to use a multipartite inequality as a seed. We chose a Bell inequality for three parties -- that witnesses GMNL in tripartite correlations -- that belongs to  class  $5$ of \cite{Bancal:PRA:014102} and that can be written as
\begin{equation}\begin{split}\label{Ifor3}
I^{A_1A_2A_3}_{\textrm{tri}} = &P(000|000) - P(010|111) - \\-&P(000|011)
 -  P(001|001) - P(100|110) -\\&- P(010|010) - P(100|100) \leq 0 
\end{split}\end{equation}
and hence satisfying condition \eqref{GG:UsefulIneqs2}. This allows us to construct, again, two new families of Bell inequalities witnessing GMNL for any $n \geq 4$.

\begin{theorem}\label{GG:TripartiteSeedTHRM}
The families of inequalities
\begin{equation}\begin{split}\label{IGeneralsym3seed}
I^{A_1\ldots A_n}_{\textrm{tri},\textrm{sym}} =\sum\limits_{i=1}^{n-2}\sum\limits_{j>i}^{n-1} \sum\limits_{k>j}^{n} I^{A_iA_jA_k}_{\textrm{tri},\vec{0}|\vec{0}} - {n-1\choose 3}  P(\vec{0}|\vec{0}) \leq 0
\end{split}\end{equation}
\begin{equation}\begin{split}\label{IGeneral3seed}
I^{A_1\ldots A_n}_{\textrm{tri},\textrm{\ding{192}\ding{193}}} = \sum\limits_{j>2}^{n} I^{A_1A_2A_j}_{\textrm{tri},\vec{0}|\vec{0}} - (n-3)  P(\vec{0}|\vec{0}) \leq 0
\end{split}\end{equation}
\end{theorem}
witness GMNL for  $n\geq 4$ parties.

\begin{proof}
Although this proof is again analogous to the previous families, a tripartite inequality as the seed changes the weights associated to the term $P(\vec0|\vec0)$. 
For the symmetric family \eqref{IGeneralsym3seed}, consider that a biseparable probability distribution \eqref{biseparable} can violate at most ${n-1\choose 3}$ inequalities $I^{A_iA_jA_k}_{\textrm{tri},\vec{0}|\vec{0}}$ between three different parties. This comes from the fact that the best for a biseparable distribution is a grouping $g=\{1,2,...,n-1\}$, $\bar{g}=\{n\}$ of the parties, allowing a maximum of ${n-1\choose 3}$ to potentially violate the inequality $I^{A_iA_jA_k}_{\textrm{tri},\vec{0}|\vec{0}}$.
For the centered family \eqref{IGeneral3seed}, consider that there are at most $n-3$ inequalities $I^{A_1A_2A_j}_{\textrm{tri},\vec{0}|\vec{0}}$ that can potentially be violated for a grouping $g=\{1,2,...,n-1\}$, $\bar{g}=\{n\}$ of the parties. 
\end{proof}

%IMAGE?\\ % (same n=5 but for 3 parts in two groups)

%The goal of this section was to provide insight on the generality of our method for the construction of Bell-like inequalities witnessing GMNL for any $n$ from a witness of GMNL for $m<n$ parties. Indeed, any seed $I^{A_1...A_m}$ can be used in a construction as ours as long as it can be written in a form such that $I^{A_1...A_m} - P(\vec{0}|\vec{0}) \leq 0$ (for any probability distribution).

It would obviously be interesting to explore up to which extent the inequalities that can be built -- as the ones in \eqref{IGeneralsymBeta},\eqref{IGeneralBeta}, \eqref{IGeneralsym3seed} and \eqref{IGeneral3seed} -- are useful to witness GMNL from quantum states. We leave this direction of research open for further work. Finally, it would also be insightful to consider seeds allowing for more measurement choices and/or outcomes.

\section{Constructing $m$-way nonlocality Bell inequalities}\label{GG:MoreKway}
Now we show how our construction also allows one to build families of Bell inequalities that witness intermediate types of multipartite nonlocality. Indeed, in the multipartite scenario it is possible to define a hierarchy of multipartite correlations taking into account the extent to which these are multipartite nonlocal. This can be measured, for example, by notions such as $m-$way (non)locality or $m-$separability of correlations \cite{Bancal:PRA:014102,BancalQuantifying}. Instead of asking whether given correlations can be decomposed into (convex mixtures) of two groups as in \eqref{biseparable}, one can ask whether the correlations can be decomposed into $m$ groups. Correlations that are decomposable into $m<n$ groups are then less multipartite nonlocal than other correlations that do not allow for such decomposition.\\

Correlations $P(\vec{a}|\vec{x})$ are said to be $m-$separable (or $m-$way local), i.e. decomposable into $m$ groups, if
\begin{equation}\begin{split}\label{mseparable}
P_{\textrm{m-sep}}(\vec{a}|\vec{x}) = \sum\limits_{k}\sum\limits_{\lambda_k} q_k(\lambda_k) \prod\limits_{i=1}^m P(\vec{a}_{k_i}|\vec{x}_{k_i},\lambda_{k})
\end{split}\end{equation}
where $\sum\limits_{k}\sum\limits_{\lambda_k} q_k(\lambda_k) = 1$, $q_k(\lambda_k) \geq 0$. Here, the variable $k$ defines a \textit{grouping} of the $n$ parties into $m$ pairwise disjoint and non-empty groups $k_i$, $i = 1,2,...m$: $|k_i|>0 \hspace{0.2cm} \forall i$, $k_i \cap k_j = \emptyset \hspace{0.2cm} \forall i \neq j$ and $\sum\limits_{i=1}^m |k_i| = n$. Biseparable correlations \eqref{biseparable} for $m=2$ can be decomposed into (convex mixtures of) two group of parties $k_1 = g$ and $k_2 = \bar{g}$.\\

For the sake of simplicity, we use the seed inequality $I^{A_1A_2}$ \eqref{IA1A2} to generalise our two families of Bell inequalities, symmetric and centered, for the detection of $m$-way nonlocality.

\begin{theorem}\label{GG:msepTHRM}
The families of inequalities for $n$ parties
\begin{equation}\begin{split}\label{ImsepSYM}
I^{A_1\ldots A_n}_{\textrm{m-sep},\textrm{sym}} =\sum\limits_{i=1}^{n-1}\sum\limits_{j>i}^{n} I^{A_iA_j}_{\vec{0}|\vec{0}} - {n+1-m\choose 2}  P(\vec{0}|\vec{0}) \leq 0
\end{split}\end{equation}
\begin{equation}\begin{split}\label{Imsep}
I^{A_1\ldots A_n}_{\textrm{m-sep},\textrm{\ding{192}}} = \sum\limits_{j>1}^{n} I^{A_1A_j}_{\textrm{m},\vec{0}|\vec{0}} - (n-m)  P(\vec{0}|\vec{0}) \leq 0
\end{split}\end{equation}
witness $m-$way nonlocality (or non $m-$separability) for all $n\geq3, m<n$.
\end{theorem}

\begin{proof}
The proofs follow the same line as the proofs for the other families of inequalities we have already constructed. By making $m$ groups instead of $2$, one needs to count the maximum number of pairs of parties $A_iA_j$ that can be made inside all the $m$ groups for the family \eqref{ImsepSYM}. Indeed, only pairs $A_iA_j$ inside the same group can potentially violate a lifted inequality $I^{A_iA_j}_{\vec{0}|\vec{0}}$. The best way to group $n$ parties into $m$ groups, in order to maximise the number of such pairs of parties, is to put $n-m+1$ parties into one group and the remaining $m-1$ ones into one group each. In this way, a maximum amount of ${n+1-m\choose 2}$ inequalities $I^{A_iA_j}_{\vec{0}|\vec{0}}$ can potentially be violated, but these can be cancelled by the ${n+1-m\choose 2}  P(\vec{0}|\vec{0})$ terms in \eqref{ImsepSYM} since $I^{A_iA_j}_{\vec{0}|\vec{0}}-P(\vec{0}|\vec{0}) \leq 0$.

For the family \eqref{Imsep}, one needs to count the maximum number of pairs $A_1A_j$ that can be made inside the group containing party $A_1$. By putting the maximal number of $n-m$ parties, plus party $A_1$, in one group, one gets that a maximum number of $n-m$ pairs $A_1A_j$ can be formed. This implies that a maximum amount of $n-m$ inequalities $I^{A_iA_j}_{\vec{0}|\vec{0}}$ can potentially be violated, but these are cancelled by the $(n-m)  P(\vec{0}|\vec{0})$ term in \eqref{Imsep} since $I^{A_iA_j}_{\vec{0}|\vec{0}} - P(\vec{0}|\vec{0}) \leq 0$.
\end{proof}

\section{Conclusion}

In this work, we have introduced a versatile technique to build Bell inequalities for the $n$-partite scenario. It consists of taking a ``seed'' -- a Bell inequality for $m<n$ parties obeying certain constraints -- to generate new families of Bell inequalities. Intuitively, the seed defines the nonlocal game that will be played by numerous groups of $m$ parties in the $n$-partite system. The specification of the sets of parties that are required to play the nonlocal game defines the level of multipartite nonlocality to be detected. Indeed, our construction can be used to witness $k$-nonlocal multipartite correlations, including the stronger notion of genuine multipartite nonlocality (where all the parties of the system are nonclassically correlated).

To illustrate the power of our construction, we have first used the CHSH inequality as the seed to build two new families of Bell inequalities, $I_\textrm{sym}$ and $I_{\textrm{\ding{192}}}$, for the detection of genuine multipartite nonlocality in systems with any number $n\geq3$ of parties. We showed that these families are particularly useful for the detection of genuine multipartite nonlocality (GMNL) in genuine multipartite entangled (GME) pure states. Indeed, for $n=3$ we proved that $I^{A_1A_2A_3}_{\textrm{\ding{192}}}$ is able to detect genuine tripartite nonlocality in all genuine entangled three-qubit pure states invariant under the permutation of two parties. We also showed that the family $I_\textrm{sym}$ witnesses GMNL in all pure GME GHZ-like states $\cos\theta\ket{0}^{\otimes n} - \sin\theta\ket{1}^{\otimes n}$, including those which are almost product. Note that for three parties, $I^{A_1A_2A_3}_\textrm{sym}$ coincides with a Bell inequality found in \cite{Bancal:PRA:014102}, where numerical evidence was given that it detects GMNL in all GME three-qubit pure states. We extended this numerical evidence to the four-partite case, using $I^{A_1\ldots A_4}_\textrm{sym}$ to detect GMNL in all pure GME four-qubit states. Taking into account these partial results and the operational meaning of the family $I_\textrm{sym}$, we conjecture that this single family of Bell inequalities can be used to show that all GME pure states display GMNL, establishing a direct relation between genuine multipartite notions of pure state entanglement and nonlocality. 

Apart from a proof in full generality, which seems not straightforward, it would be interesting to extend these results to more families of GME pure states. A possibility is to study the multipartite $W$-state, $\ket{W_n}=\frac{1}{\sqrt{n}} \sum\limits_{i=1}^n \ket{1}_i \otimes_{j \neq i} \ket{0}_j$, for which we already managed to obtain, numerically, violations up to $n=5$. It would also be interesting to use further the characterisation of all three-qubit systems in a pure state ~\cite{AcinThreeQubits} to obtain an analytical proof that our inequality for $n=3$ detects GMNL when these states are GME by detecting also the states which have no symmetries. Further numerical exploration, for more observers or systems of larger dimensions, is another possibility.

We have further used our technique to build families of Bell inequalities taking two different seeds: the bipartite ``tilted CHSH inequality'' and a tripartite inequality witnessing GMNL from \cite{Bancal:PRA:014102}. We have also shown how to design families of Bell inequalities that detect $m$-way multipartite nonlocality. In all cases, the construction is quite straightforward, which shows the potential and versatility of the method.

For future work, it would be interesting to further explore the applicability of the Bell inequalities built through our method. For instance, can the family of Bell inequalities with the tilted CHSH inequality as a seed \eqref{Ibeta} be used to self-test certain classes of multipartite entangled pure states? %Can we establish a direct relation between $k$-way entanglement in pure states and $k$-way nonlocality?
Also generalising our seed to more settings and/or outcomes has the potential for generating Bell inequalities fit to detect multipartite correlations in a whole range of new different scenarios. 

%it would be interesting to explore the strength of the witnesses of $m-$way nonlocality of sec.\ref{GG:MoreKway} by, for example, using numerical techniques to obtain violation from states that are not GME. More generally, it would be interesting to understand the link between $m-$separability (of states) and $m-$way nonlocality, both for pure and mixed states.\\

%Is it possible to obtain an even more general result by showing that all states that are not $m-$separable -- i.e. that are entangled across $m+1$ groups of parties at least -- can exhibit genuine $m-$way nonlocality, as witnessed by a violation of our families of inequalities?\\

%Finally, we have seen that our method to construct families of Bell inequalities is very fertile and can be used to obtain many other families witnessing multipartite nonlocality. We also expect that our families of inequalities with various seeds, such as the tilted CHSH inequality  \eqref{Ibeta}, can be used to self-test certain classes of multipartite entangled pure states.

%
\begin{acknowledgments}
This work  is  supported  by  the  ERC  CoG  QITBOX,  the AXA Chair in Quantum Information Science, the Templeton Foundation,  the  Spanish  MINECO  (QIBEQI  FIS2016-80773-P  and  Severo Ochoa SEV-2015-0522), Fundacio Cellex, Generalitat de Catalunya (CERCA Program and SGR1381).
\end{acknowledgments}

%INITIAL BIBLIO:

%\bibliographystyle{GenuineGisin}{unsrt}
%\bibliography{Biblio_GenuineGisin}  %REM: To be uncommented for standard biblio.

\begin{thebibliography}{41}%
\makeatletter
\providecommand \@ifxundefined [1]{%
 \@ifx{#1\undefined}
}%
\providecommand \@ifnum [1]{%
 \ifnum #1\expandafter \@firstoftwo
 \else \expandafter \@secondoftwo
 \fi
}%
\providecommand \@ifx [1]{%
 \ifx #1\expandafter \@firstoftwo
 \else \expandafter \@secondoftwo
 \fi
}%
\providecommand \natexlab [1]{#1}%
\providecommand \enquote  [1]{``#1''}%
\providecommand \bibnamefont  [1]{#1}%
\providecommand \bibfnamefont [1]{#1}%
\providecommand \citenamefont [1]{#1}%
\providecommand \href@noop [0]{\@secondoftwo}%
\providecommand \href [0]{\begingroup \@sanitize@url \@href}%
\providecommand \@href[1]{\@@startlink{#1}\@@href}%
\providecommand \@@href[1]{\endgroup#1\@@endlink}%
\providecommand \@sanitize@url [0]{\catcode `\\12\catcode `\$12\catcode
  `\&12\catcode `\#12\catcode `\^12\catcode `\_12\catcode `\%12\relax}%
\providecommand \@@startlink[1]{}%
\providecommand \@@endlink[0]{}%
\providecommand \url  [0]{\begingroup\@sanitize@url \@url }%
\providecommand \@url [1]{\endgroup\@href {#1}{\urlprefix }}%
\providecommand \urlprefix  [0]{URL }%
\providecommand \Eprint [0]{\href }%
\providecommand \doibase [0]{http://dx.doi.org/}%
\providecommand \selectlanguage [0]{\@gobble}%
\providecommand \bibinfo  [0]{\@secondoftwo}%
\providecommand \bibfield  [0]{\@secondoftwo}%
\providecommand \translation [1]{[#1]}%
\providecommand \BibitemOpen [0]{}%
\providecommand \bibitemStop [0]{}%
\providecommand \bibitemNoStop [0]{.\EOS\space}%
\providecommand \EOS [0]{\spacefactor3000\relax}%
\providecommand \BibitemShut  [1]{\csname bibitem#1\endcsname}%
\let\auto@bib@innerbib\@empty
%</preamble>
\bibitem [{\citenamefont {Einstein}\ \emph {et~al.}(1935)\citenamefont
  {Einstein}, \citenamefont {Podolsky},\ and\ \citenamefont {Rosen}}]{EPR}%
  \BibitemOpen
  \bibfield  {author} {\bibinfo {author} {\bibfnamefont {A.}~\bibnamefont
  {Einstein}}, \bibinfo {author} {\bibfnamefont {B.}~\bibnamefont {Podolsky}},
  \ and\ \bibinfo {author} {\bibfnamefont {N.}~\bibnamefont {Rosen}},\
  }\href@noop {} {\bibfield  {journal} {\bibinfo  {journal} {Physical review}\
  }\textbf {\bibinfo {volume} {47}},\ \bibinfo {pages} {777} (\bibinfo {year}
  {1935})}\BibitemShut {NoStop}%
\bibitem [{\citenamefont {Bell}(2001)}]{Bell}%
  \BibitemOpen
  \bibfield  {author} {\bibinfo {author} {\bibfnamefont {J.~S.}\ \bibnamefont
  {Bell}},\ }in\ \href@noop {} {\emph {\bibinfo {booktitle} {John S Bell On The
  Foundations Of Quantum Mechanics}}}\ (\bibinfo  {publisher} {World
  Scientific},\ \bibinfo {year} {2001})\ pp.\ \bibinfo {pages}
  {7--12}\BibitemShut {NoStop}%
\bibitem [{\citenamefont {Colbeck}(2009)}]{colbeck2009quantum}%
  \BibitemOpen
  \bibfield  {author} {\bibinfo {author} {\bibfnamefont {R.}~\bibnamefont
  {Colbeck}},\ }\href@noop {} {\bibfield  {journal} {\bibinfo  {journal} {arXiv
  preprint arXiv:0911.3814}\ } (\bibinfo {year} {2009})}\BibitemShut {NoStop}%
\bibitem [{\citenamefont {Colbeck}\ and\ \citenamefont
  {Kent}(2011)}]{ColbeckExpansion}%
  \BibitemOpen
  \bibfield  {author} {\bibinfo {author} {\bibfnamefont {R.}~\bibnamefont
  {Colbeck}}\ and\ \bibinfo {author} {\bibfnamefont {A.}~\bibnamefont {Kent}},\
  }\href@noop {} {\bibfield  {journal} {\bibinfo  {journal} {Journal of Physics
  A: Mathematical and Theoretical}\ }\textbf {\bibinfo {volume} {44}},\
  \bibinfo {pages} {095305} (\bibinfo {year} {2011})}\BibitemShut {NoStop}%
\bibitem [{\citenamefont {Pironio}\ \emph {et~al.}(2010)\citenamefont
  {Pironio}, \citenamefont {Ac{\'i}n}, \citenamefont {Massar}, \citenamefont
  {de~La~Giroday}, \citenamefont {Matsukevich}, \citenamefont {Maunz},
  \citenamefont {Olmschenk}, \citenamefont {Hayes}, \citenamefont {Luo},
  \citenamefont {Manning},\ and\ \citenamefont {Monroe}}]{Randomness}%
  \BibitemOpen
  \bibfield  {author} {\bibinfo {author} {\bibfnamefont {S.}~\bibnamefont
  {Pironio}}, \bibinfo {author} {\bibfnamefont {A.}~\bibnamefont {Ac{\'i}n}},
  \bibinfo {author} {\bibfnamefont {S.}~\bibnamefont {Massar}}, \bibinfo
  {author} {\bibfnamefont {A.~B.}\ \bibnamefont {de~La~Giroday}}, \bibinfo
  {author} {\bibfnamefont {D.~N.}\ \bibnamefont {Matsukevich}}, \bibinfo
  {author} {\bibfnamefont {P.}~\bibnamefont {Maunz}}, \bibinfo {author}
  {\bibfnamefont {S.}~\bibnamefont {Olmschenk}}, \bibinfo {author}
  {\bibfnamefont {D.}~\bibnamefont {Hayes}}, \bibinfo {author} {\bibfnamefont
  {L.}~\bibnamefont {Luo}}, \bibinfo {author} {\bibfnamefont {T.~A.}\
  \bibnamefont {Manning}}, \ and\ \bibinfo {author} {\bibfnamefont
  {C.}~\bibnamefont {Monroe}},\ }\href@noop {} {\bibfield  {journal} {\bibinfo
  {journal} {Nature}\ }\textbf {\bibinfo {volume} {464}},\ \bibinfo {pages}
  {1021} (\bibinfo {year} {2010})}\BibitemShut {NoStop}%
\bibitem [{\citenamefont {Colbeck}\ and\ \citenamefont
  {Renner}(2012)}]{ColbeckAmplification}%
  \BibitemOpen
  \bibfield  {author} {\bibinfo {author} {\bibfnamefont {R.}~\bibnamefont
  {Colbeck}}\ and\ \bibinfo {author} {\bibfnamefont {R.}~\bibnamefont
  {Renner}},\ }\href@noop {} {\bibfield  {journal} {\bibinfo  {journal} {Nature
  Physics}\ }\textbf {\bibinfo {volume} {8}},\ \bibinfo {pages} {450} (\bibinfo
  {year} {2012})}\BibitemShut {NoStop}%
\bibitem [{\citenamefont {Ekert}(1991)}]{EkertQKD}%
  \BibitemOpen
  \bibfield  {author} {\bibinfo {author} {\bibfnamefont {A.~K.}\ \bibnamefont
  {Ekert}},\ }\href@noop {} {\bibfield  {journal} {\bibinfo  {journal}
  {Physical review letters}\ }\textbf {\bibinfo {volume} {67}},\ \bibinfo
  {pages} {661} (\bibinfo {year} {1991})}\BibitemShut {NoStop}%
\bibitem [{\citenamefont {Barrett}\ \emph {et~al.}(2005)\citenamefont
  {Barrett}, \citenamefont {Hardy},\ and\ \citenamefont
  {Kent}}]{barrett2005no}%
  \BibitemOpen
  \bibfield  {author} {\bibinfo {author} {\bibfnamefont {J.}~\bibnamefont
  {Barrett}}, \bibinfo {author} {\bibfnamefont {L.}~\bibnamefont {Hardy}}, \
  and\ \bibinfo {author} {\bibfnamefont {A.}~\bibnamefont {Kent}},\ }\href@noop
  {} {\bibfield  {journal} {\bibinfo  {journal} {Physical review letters}\
  }\textbf {\bibinfo {volume} {95}},\ \bibinfo {pages} {010503} (\bibinfo
  {year} {2005})}\BibitemShut {NoStop}%
\bibitem [{\citenamefont {Ac{\'\i}n}\ \emph {et~al.}(2007)\citenamefont
  {Ac{\'\i}n}, \citenamefont {Brunner}, \citenamefont {Gisin}, \citenamefont
  {Massar}, \citenamefont {Pironio},\ and\ \citenamefont
  {Scarani}}]{acin2007device}%
  \BibitemOpen
  \bibfield  {author} {\bibinfo {author} {\bibfnamefont {A.}~\bibnamefont
  {Ac{\'\i}n}}, \bibinfo {author} {\bibfnamefont {N.}~\bibnamefont {Brunner}},
  \bibinfo {author} {\bibfnamefont {N.}~\bibnamefont {Gisin}}, \bibinfo
  {author} {\bibfnamefont {S.}~\bibnamefont {Massar}}, \bibinfo {author}
  {\bibfnamefont {S.}~\bibnamefont {Pironio}}, \ and\ \bibinfo {author}
  {\bibfnamefont {V.}~\bibnamefont {Scarani}},\ }\href@noop {} {\bibfield
  {journal} {\bibinfo  {journal} {Physical Review Letters}\ }\textbf {\bibinfo
  {volume} {98}},\ \bibinfo {pages} {230501} (\bibinfo {year}
  {2007})}\BibitemShut {NoStop}%
\bibitem [{\citenamefont {Mayers}\ and\ \citenamefont
  {Yao}(1998)}]{MayersSelfTesting}%
  \BibitemOpen
  \bibfield  {author} {\bibinfo {author} {\bibfnamefont {D.}~\bibnamefont
  {Mayers}}\ and\ \bibinfo {author} {\bibfnamefont {A.}~\bibnamefont {Yao}},\
  }in\ \href@noop {} {\emph {\bibinfo {booktitle} {Foundations of Computer
  Science, 1998. Proceedings. 39th Annual Symposium on}}}\ (\bibinfo
  {organization} {IEEE},\ \bibinfo {year} {1998})\ pp.\ \bibinfo {pages}
  {503--509}\BibitemShut {NoStop}%
\bibitem [{\citenamefont {Svetlichny}(1987)}]{Svetlichny}%
  \BibitemOpen
  \bibfield  {author} {\bibinfo {author} {\bibfnamefont {G.}~\bibnamefont
  {Svetlichny}},\ }\href@noop {} {\bibfield  {journal} {\bibinfo  {journal}
  {Physical Review D}\ }\textbf {\bibinfo {volume} {35}},\ \bibinfo {pages}
  {3066} (\bibinfo {year} {1987})}\BibitemShut {NoStop}%
\bibitem [{\citenamefont {Seevinck}\ and\ \citenamefont
  {Svetlichny}(2002)}]{Svetlichny2}%
  \BibitemOpen
  \bibfield  {author} {\bibinfo {author} {\bibfnamefont {M.}~\bibnamefont
  {Seevinck}}\ and\ \bibinfo {author} {\bibfnamefont {G.}~\bibnamefont
  {Svetlichny}},\ }\href@noop {} {\bibfield  {journal} {\bibinfo  {journal}
  {Physical review letters}\ }\textbf {\bibinfo {volume} {89}},\ \bibinfo
  {pages} {060401} (\bibinfo {year} {2002})}\BibitemShut {NoStop}%
\bibitem [{\citenamefont {Gallego}\ \emph {et~al.}(2012)\citenamefont
  {Gallego}, \citenamefont {W{\"u}rflinger}, \citenamefont {Ac{\'\i}n},\ and\
  \citenamefont {Navascu{\'e}s}}]{Gallego:PRL:070401}%
  \BibitemOpen
  \bibfield  {author} {\bibinfo {author} {\bibfnamefont {R.}~\bibnamefont
  {Gallego}}, \bibinfo {author} {\bibfnamefont {L.~E.}\ \bibnamefont
  {W{\"u}rflinger}}, \bibinfo {author} {\bibfnamefont {A.}~\bibnamefont
  {Ac{\'\i}n}}, \ and\ \bibinfo {author} {\bibfnamefont {M.}~\bibnamefont
  {Navascu{\'e}s}},\ }\href@noop {} {\bibfield  {journal} {\bibinfo  {journal}
  {Physical review letters}\ }\textbf {\bibinfo {volume} {109}},\ \bibinfo
  {pages} {070401} (\bibinfo {year} {2012})}\BibitemShut {NoStop}%
\bibitem [{\citenamefont {Bancal}\ \emph {et~al.}(2013)\citenamefont {Bancal},
  \citenamefont {Barrett}, \citenamefont {Gisin},\ and\ \citenamefont
  {Pironio}}]{Bancal:PRA:014102}%
  \BibitemOpen
  \bibfield  {author} {\bibinfo {author} {\bibfnamefont {J.-D.}\ \bibnamefont
  {Bancal}}, \bibinfo {author} {\bibfnamefont {J.}~\bibnamefont {Barrett}},
  \bibinfo {author} {\bibfnamefont {N.}~\bibnamefont {Gisin}}, \ and\ \bibinfo
  {author} {\bibfnamefont {S.}~\bibnamefont {Pironio}},\ }\href@noop {}
  {\bibfield  {journal} {\bibinfo  {journal} {Physical Review A}\ }\textbf
  {\bibinfo {volume} {88}},\ \bibinfo {pages} {014102} (\bibinfo {year}
  {2013})}\BibitemShut {NoStop}%
\bibitem [{\citenamefont {Bancal}(2014{\natexlab{a}})}]{BancalDIEW}%
  \BibitemOpen
  \bibfield  {author} {\bibinfo {author} {\bibfnamefont {J.-D.}\ \bibnamefont
  {Bancal}},\ }in\ \href@noop {} {\emph {\bibinfo {booktitle} {On the
  Device-Independent Approach to Quantum Physics}}}\ (\bibinfo  {publisher}
  {Springer},\ \bibinfo {year} {2014})\ pp.\ \bibinfo {pages}
  {73--80}\BibitemShut {NoStop}%
\bibitem [{\citenamefont {Gallego}\ \emph {et~al.}(2013)\citenamefont
  {Gallego}, \citenamefont {Masanes}, \citenamefont {De~La~Torre},
  \citenamefont {Dhara}, \citenamefont {Aolita},\ and\ \citenamefont
  {Ac{\'\i}n}}]{gallego2013full}%
  \BibitemOpen
  \bibfield  {author} {\bibinfo {author} {\bibfnamefont {R.}~\bibnamefont
  {Gallego}}, \bibinfo {author} {\bibfnamefont {L.}~\bibnamefont {Masanes}},
  \bibinfo {author} {\bibfnamefont {G.}~\bibnamefont {De~La~Torre}}, \bibinfo
  {author} {\bibfnamefont {C.}~\bibnamefont {Dhara}}, \bibinfo {author}
  {\bibfnamefont {L.}~\bibnamefont {Aolita}}, \ and\ \bibinfo {author}
  {\bibfnamefont {A.}~\bibnamefont {Ac{\'\i}n}},\ }\href@noop {} {\bibfield
  {journal} {\bibinfo  {journal} {Nature communications}\ }\textbf {\bibinfo
  {volume} {4}},\ \bibinfo {pages} {2654} (\bibinfo {year} {2013})}\BibitemShut
  {NoStop}%
\bibitem [{\citenamefont {Bancal}(2014{\natexlab{b}})}]{bancal2014quantum}%
  \BibitemOpen
  \bibfield  {author} {\bibinfo {author} {\bibfnamefont {J.-D.}\ \bibnamefont
  {Bancal}},\ }in\ \href@noop {} {\emph {\bibinfo {booktitle} {On the
  Device-Independent Approach to Quantum Physics}}}\ (\bibinfo  {publisher}
  {Springer},\ \bibinfo {year} {2014})\ pp.\ \bibinfo {pages}
  {97--105}\BibitemShut {NoStop}%
\bibitem [{\citenamefont {Bouda}\ \emph {et~al.}(2014)\citenamefont {Bouda},
  \citenamefont {Paw{\l}owski}, \citenamefont {Pivoluska},\ and\ \citenamefont
  {Plesch}}]{BoudaRandomness}%
  \BibitemOpen
  \bibfield  {author} {\bibinfo {author} {\bibfnamefont {J.}~\bibnamefont
  {Bouda}}, \bibinfo {author} {\bibfnamefont {M.}~\bibnamefont {Paw{\l}owski}},
  \bibinfo {author} {\bibfnamefont {M.}~\bibnamefont {Pivoluska}}, \ and\
  \bibinfo {author} {\bibfnamefont {M.}~\bibnamefont {Plesch}},\ }\href@noop {}
  {\bibfield  {journal} {\bibinfo  {journal} {Physical Review A}\ }\textbf
  {\bibinfo {volume} {90}},\ \bibinfo {pages} {032313} (\bibinfo {year}
  {2014})}\BibitemShut {NoStop}%
\bibitem [{\citenamefont {Tura}\ \emph {et~al.}(2014)\citenamefont {Tura},
  \citenamefont {Augusiak}, \citenamefont {Sainz}, \citenamefont {V{\'e}rtesi},
  \citenamefont {Lewenstein},\ and\ \citenamefont {Ac{\'\i}n}}]{Jordi}%
  \BibitemOpen
  \bibfield  {author} {\bibinfo {author} {\bibfnamefont {J.}~\bibnamefont
  {Tura}}, \bibinfo {author} {\bibfnamefont {R.}~\bibnamefont {Augusiak}},
  \bibinfo {author} {\bibfnamefont {A.~B.}\ \bibnamefont {Sainz}}, \bibinfo
  {author} {\bibfnamefont {T.}~\bibnamefont {V{\'e}rtesi}}, \bibinfo {author}
  {\bibfnamefont {M.}~\bibnamefont {Lewenstein}}, \ and\ \bibinfo {author}
  {\bibfnamefont {A.}~\bibnamefont {Ac{\'\i}n}},\ }\href@noop {} {\bibfield
  {journal} {\bibinfo  {journal} {Science}\ }\textbf {\bibinfo {volume}
  {344}},\ \bibinfo {pages} {1256} (\bibinfo {year} {2014})}\BibitemShut
  {NoStop}%
\bibitem [{\citenamefont {Bancal}\ \emph {et~al.}(2009)\citenamefont {Bancal},
  \citenamefont {Branciard}, \citenamefont {Gisin},\ and\ \citenamefont
  {Pironio}}]{BancalQuantifying}%
  \BibitemOpen
  \bibfield  {author} {\bibinfo {author} {\bibfnamefont {J.-D.}\ \bibnamefont
  {Bancal}}, \bibinfo {author} {\bibfnamefont {C.}~\bibnamefont {Branciard}},
  \bibinfo {author} {\bibfnamefont {N.}~\bibnamefont {Gisin}}, \ and\ \bibinfo
  {author} {\bibfnamefont {S.}~\bibnamefont {Pironio}},\ }\href@noop {}
  {\bibfield  {journal} {\bibinfo  {journal} {Physical Review Letters}\
  }\textbf {\bibinfo {volume} {103}},\ \bibinfo {pages} {090503} (\bibinfo
  {year} {2009})}\BibitemShut {NoStop}%
\bibitem [{\citenamefont {Werner}(1989)}]{Werner}%
  \BibitemOpen
  \bibfield  {author} {\bibinfo {author} {\bibfnamefont {R.~F.}\ \bibnamefont
  {Werner}},\ }\href@noop {} {\bibfield  {journal} {\bibinfo  {journal}
  {Physical Review A}\ }\textbf {\bibinfo {volume} {40}},\ \bibinfo {pages}
  {4277} (\bibinfo {year} {1989})}\BibitemShut {NoStop}%
\bibitem [{\citenamefont {Barrett}(2002)}]{BarrettPOVM}%
  \BibitemOpen
  \bibfield  {author} {\bibinfo {author} {\bibfnamefont {J.}~\bibnamefont
  {Barrett}},\ }\href@noop {} {\bibfield  {journal} {\bibinfo  {journal}
  {Physical Review A}\ }\textbf {\bibinfo {volume} {65}},\ \bibinfo {pages}
  {042302} (\bibinfo {year} {2002})}\BibitemShut {NoStop}%
\bibitem [{\citenamefont {Gisin}(1991)}]{GisinThrm}%
  \BibitemOpen
  \bibfield  {author} {\bibinfo {author} {\bibfnamefont {N.}~\bibnamefont
  {Gisin}},\ }\href@noop {} {\bibfield  {journal} {\bibinfo  {journal} {Physics
  Letters A}\ }\textbf {\bibinfo {volume} {154}},\ \bibinfo {pages} {201}
  (\bibinfo {year} {1991})}\BibitemShut {NoStop}%
\bibitem [{\citenamefont {Popescu}\ and\ \citenamefont
  {Rohrlich}(1992)}]{PopescuGeneric}%
  \BibitemOpen
  \bibfield  {author} {\bibinfo {author} {\bibfnamefont {S.}~\bibnamefont
  {Popescu}}\ and\ \bibinfo {author} {\bibfnamefont {D.}~\bibnamefont
  {Rohrlich}},\ }\href@noop {} {\bibfield  {journal} {\bibinfo  {journal}
  {Physics Letters A}\ }\textbf {\bibinfo {volume} {166}},\ \bibinfo {pages}
  {293} (\bibinfo {year} {1992})}\BibitemShut {NoStop}%
\bibitem [{\citenamefont {Gachechiladze}\ and\ \citenamefont
  {G{\"u}hne}(2016)}]{PopescuErratum}%
  \BibitemOpen
  \bibfield  {author} {\bibinfo {author} {\bibfnamefont {M.}~\bibnamefont
  {Gachechiladze}}\ and\ \bibinfo {author} {\bibfnamefont {O.}~\bibnamefont
  {G{\"u}hne}},\ }\href@noop {} {\bibfield  {journal} {\bibinfo  {journal}
  {arXiv preprint arXiv:1607.02948}\ } (\bibinfo {year} {2016})}\BibitemShut
  {NoStop}%
\bibitem [{\citenamefont {Yu}\ and\ \citenamefont
  {Oh}(2013{\natexlab{a}})}]{Yu3}%
  \BibitemOpen
  \bibfield  {author} {\bibinfo {author} {\bibfnamefont {S.}~\bibnamefont
  {Yu}}\ and\ \bibinfo {author} {\bibfnamefont {C.}~\bibnamefont {Oh}},\
  }\href@noop {} {\bibfield  {journal} {\bibinfo  {journal} {arXiv preprint
  arXiv:1306.5330}\ } (\bibinfo {year} {2013}{\natexlab{a}})}\BibitemShut
  {NoStop}%
\bibitem [{\citenamefont {Chen}\ \emph {et~al.}(2014)\citenamefont {Chen},
  \citenamefont {Yu}, \citenamefont {Zhang}, \citenamefont {Lai},\ and\
  \citenamefont {Oh}}]{Chinese1}%
  \BibitemOpen
  \bibfield  {author} {\bibinfo {author} {\bibfnamefont {Q.}~\bibnamefont
  {Chen}}, \bibinfo {author} {\bibfnamefont {S.}~\bibnamefont {Yu}}, \bibinfo
  {author} {\bibfnamefont {C.}~\bibnamefont {Zhang}}, \bibinfo {author}
  {\bibfnamefont {C.}~\bibnamefont {Lai}}, \ and\ \bibinfo {author}
  {\bibfnamefont {C.}~\bibnamefont {Oh}},\ }\href@noop {} {\bibfield  {journal}
  {\bibinfo  {journal} {Physical review letters}\ }\textbf {\bibinfo {volume}
  {112}},\ \bibinfo {pages} {140404} (\bibinfo {year} {2014})}\BibitemShut
  {NoStop}%
\bibitem [{\citenamefont {Hardy}(1993)}]{Hardy}%
  \BibitemOpen
  \bibfield  {author} {\bibinfo {author} {\bibfnamefont {L.}~\bibnamefont
  {Hardy}},\ }\href@noop {} {\bibfield  {journal} {\bibinfo  {journal}
  {Physical Review Letters}\ }\textbf {\bibinfo {volume} {71}},\ \bibinfo
  {pages} {1665} (\bibinfo {year} {1993})}\BibitemShut {NoStop}%
\bibitem [{\citenamefont {Clauser}\ \emph {et~al.}(1969)\citenamefont
  {Clauser}, \citenamefont {Horne}, \citenamefont {Shimony},\ and\
  \citenamefont {Holt}}]{CHSH}%
  \BibitemOpen
  \bibfield  {author} {\bibinfo {author} {\bibfnamefont {J.~F.}\ \bibnamefont
  {Clauser}}, \bibinfo {author} {\bibfnamefont {M.~A.}\ \bibnamefont {Horne}},
  \bibinfo {author} {\bibfnamefont {A.}~\bibnamefont {Shimony}}, \ and\
  \bibinfo {author} {\bibfnamefont {R.~A.}\ \bibnamefont {Holt}},\ }\href@noop
  {} {\bibfield  {journal} {\bibinfo  {journal} {Physical review letters}\
  }\textbf {\bibinfo {volume} {23}},\ \bibinfo {pages} {880} (\bibinfo {year}
  {1969})}\BibitemShut {NoStop}%
\bibitem [{\citenamefont {Ghirardi}\ \emph {et~al.}(1980)\citenamefont
  {Ghirardi}, \citenamefont {Rimini},\ and\ \citenamefont
  {Weber}}]{GhirardiNS}%
  \BibitemOpen
  \bibfield  {author} {\bibinfo {author} {\bibfnamefont {G.-C.}\ \bibnamefont
  {Ghirardi}}, \bibinfo {author} {\bibfnamefont {A.}~\bibnamefont {Rimini}}, \
  and\ \bibinfo {author} {\bibfnamefont {T.}~\bibnamefont {Weber}},\
  }\href@noop {} {\bibfield  {journal} {\bibinfo  {journal} {Lettere Al Nuovo
  Cimento (1971-1985)}\ }\textbf {\bibinfo {volume} {27}},\ \bibinfo {pages}
  {293} (\bibinfo {year} {1980})}\BibitemShut {NoStop}%
\bibitem [{\citenamefont {Popescu}\ and\ \citenamefont {Rohrlich}(1994)}]{NS}%
  \BibitemOpen
  \bibfield  {author} {\bibinfo {author} {\bibfnamefont {S.}~\bibnamefont
  {Popescu}}\ and\ \bibinfo {author} {\bibfnamefont {D.}~\bibnamefont
  {Rohrlich}},\ }\href@noop {} {\bibfield  {journal} {\bibinfo  {journal}
  {Foundations of Physics}\ }\textbf {\bibinfo {volume} {24}},\ \bibinfo
  {pages} {379} (\bibinfo {year} {1994})}\BibitemShut {NoStop}%
\bibitem [{\citenamefont {Froissart}(1981)}]{FroissartSingleCHSH}%
  \BibitemOpen
  \bibfield  {author} {\bibinfo {author} {\bibfnamefont {M.}~\bibnamefont
  {Froissart}},\ }\href@noop {} {\bibfield  {journal} {\bibinfo  {journal} {Il
  Nuovo Cimento B (1971-1996)}\ }\textbf {\bibinfo {volume} {64}},\ \bibinfo
  {pages} {241} (\bibinfo {year} {1981})}\BibitemShut {NoStop}%
\bibitem [{\citenamefont {Fine}(1982)}]{FineSingleCHSH}%
  \BibitemOpen
  \bibfield  {author} {\bibinfo {author} {\bibfnamefont {A.}~\bibnamefont
  {Fine}},\ }\href@noop {} {\bibfield  {journal} {\bibinfo  {journal} {Physical
  Review Letters}\ }\textbf {\bibinfo {volume} {48}},\ \bibinfo {pages} {291}
  (\bibinfo {year} {1982})}\BibitemShut {NoStop}%
\bibitem [{\citenamefont {Pironio}(2005)}]{Pironio:Lifting}%
  \BibitemOpen
  \bibfield  {author} {\bibinfo {author} {\bibfnamefont {S.}~\bibnamefont
  {Pironio}},\ }\href@noop {} {\bibfield  {journal} {\bibinfo  {journal}
  {Journal of mathematical physics}\ }\textbf {\bibinfo {volume} {46}},\
  \bibinfo {pages} {062112} (\bibinfo {year} {2005})}\BibitemShut {NoStop}%
\bibitem [{\citenamefont {Ac{\'\i}n}\ \emph {et~al.}(2000)\citenamefont
  {Ac{\'\i}n}, \citenamefont {Andrianov}, \citenamefont {Costa}, \citenamefont
  {Jan{\'e}}, \citenamefont {Latorre},\ and\ \citenamefont
  {Tarrach}}]{Acin2000}%
  \BibitemOpen
  \bibfield  {author} {\bibinfo {author} {\bibfnamefont {A.}~\bibnamefont
  {Ac{\'\i}n}}, \bibinfo {author} {\bibfnamefont {A.}~\bibnamefont
  {Andrianov}}, \bibinfo {author} {\bibfnamefont {L.}~\bibnamefont {Costa}},
  \bibinfo {author} {\bibfnamefont {E.}~\bibnamefont {Jan{\'e}}}, \bibinfo
  {author} {\bibfnamefont {J.}~\bibnamefont {Latorre}}, \ and\ \bibinfo
  {author} {\bibfnamefont {R.}~\bibnamefont {Tarrach}},\ }\href@noop {}
  {\bibfield  {journal} {\bibinfo  {journal} {Physical Review Letters}\
  }\textbf {\bibinfo {volume} {85}},\ \bibinfo {pages} {1560} (\bibinfo {year}
  {2000})}\BibitemShut {NoStop}%
\bibitem [{\citenamefont {Greenberger}\ \emph {et~al.}(1989)\citenamefont
  {Greenberger}, \citenamefont {Horne},\ and\ \citenamefont
  {Zeilinger}}]{greenberger1989going}%
  \BibitemOpen
  \bibfield  {author} {\bibinfo {author} {\bibfnamefont {D.~M.}\ \bibnamefont
  {Greenberger}}, \bibinfo {author} {\bibfnamefont {M.~A.}\ \bibnamefont
  {Horne}}, \ and\ \bibinfo {author} {\bibfnamefont {A.}~\bibnamefont
  {Zeilinger}},\ }in\ \href@noop {} {\emph {\bibinfo {booktitle} {Bell{'}s
  theorem, quantum theory and conceptions of the universe}}}\ (\bibinfo
  {publisher} {Springer},\ \bibinfo {year} {1989})\ pp.\ \bibinfo {pages}
  {69--72}\BibitemShut {NoStop}%
\bibitem [{\citenamefont {Yu}\ and\ \citenamefont
  {Oh}(2013{\natexlab{b}})}]{Chinese2}%
  \BibitemOpen
  \bibfield  {author} {\bibinfo {author} {\bibfnamefont {S.}~\bibnamefont
  {Yu}}\ and\ \bibinfo {author} {\bibfnamefont {C.}~\bibnamefont {Oh}},\
  }\href@noop {} {\bibfield  {journal} {\bibinfo  {journal} {arXiv preprint
  arXiv:1306.5330}\ } (\bibinfo {year} {2013}{\natexlab{b}})}\BibitemShut
  {NoStop}%
\bibitem [{\citenamefont {Ac\'in}\ \emph {et~al.}(2012)\citenamefont {Ac\'in},
  \citenamefont {Massar},\ and\ \citenamefont {Pironio}}]{Acin2012}%
  \BibitemOpen
  \bibfield  {author} {\bibinfo {author} {\bibfnamefont {A.}~\bibnamefont
  {Ac\'in}}, \bibinfo {author} {\bibfnamefont {S.}~\bibnamefont {Massar}}, \
  and\ \bibinfo {author} {\bibfnamefont {S.}~\bibnamefont {Pironio}},\
  }\href@noop {} {\bibfield  {journal} {\bibinfo  {journal} {Phys. Rev. Lett.}\
  }\textbf {\bibinfo {volume} {108}},\ \bibinfo {pages} {100402} (\bibinfo
  {year} {2012})}\BibitemShut {NoStop}%
\bibitem [{\citenamefont {Bamps}\ and\ \citenamefont
  {Pironio}(2015)}]{bamps2015sum}%
  \BibitemOpen
  \bibfield  {author} {\bibinfo {author} {\bibfnamefont {C.}~\bibnamefont
  {Bamps}}\ and\ \bibinfo {author} {\bibfnamefont {S.}~\bibnamefont
  {Pironio}},\ }\href@noop {} {\bibfield  {journal} {\bibinfo  {journal}
  {Physical Review A}\ }\textbf {\bibinfo {volume} {91}},\ \bibinfo {pages}
  {052111} (\bibinfo {year} {2015})}\BibitemShut {NoStop}%
\bibitem [{\citenamefont {Acin}\ \emph {et~al.}(2001)\citenamefont {Acin},
  \citenamefont {Andrianov}, \citenamefont {Jan{\'e}},\ and\ \citenamefont
  {Tarrach}}]{AcinThreeQubits}%
  \BibitemOpen
  \bibfield  {author} {\bibinfo {author} {\bibfnamefont {A.}~\bibnamefont
  {Acin}}, \bibinfo {author} {\bibfnamefont {A.}~\bibnamefont {Andrianov}},
  \bibinfo {author} {\bibfnamefont {E.}~\bibnamefont {Jan{\'e}}}, \ and\
  \bibinfo {author} {\bibfnamefont {R.}~\bibnamefont {Tarrach}},\ }\href@noop
  {} {\bibfield  {journal} {\bibinfo  {journal} {Journal of Physics A:
  Mathematical and General}\ }\textbf {\bibinfo {volume} {34}},\ \bibinfo
  {pages} {6725} (\bibinfo {year} {2001})}\BibitemShut {NoStop}%
\bibitem [{Note1()}]{Note1}%
  \BibitemOpen
  \bibinfo {note} {Remark additionally that we do not make use of a potential
  second degree of freedom (the phase).}\BibitemShut {Stop}%
\end{thebibliography}
%\begin{thebibliography}{9}%
%\bibitem{Flo} 
%\end{thebibliography}%

%---------------------------------------------------------------------------------
%% BIBLIO FOR ARXIV:


%---------------------------------------------------------------------------------

\clearpage

\onecolumngrid
\appendix

\section{Lifting Bell inequalities to more observers}\label{APP1:liftings}The technique of lifting a Bell inequality consists in taking an inequality designed for a specific Bell set-up -- with a fixed number of observers, measurements and outcomes-- and extending it to a set-up with an increased number of any of these variables. Here we are interested in lifting a Bell inequality to more observers. We will briefly review its definition and prove one property of these inequalities, which is used in Theorems \ref{3partyisGMNL} and \ref{NpartyisGMNL}.\\

Consider a Bell inequality for two observers \eqref{Bellineq} that, without loss of generality, can be written as
\begin{equation}\label{Bellineq2}
\mathcal{I}=\sum_{a_1a_2x_1x_2}c_{a_1a_2}^{x_1x_2}P(a_1a_2|x_1x_2)\leq 0
\end{equation}
where observer $A_i$ performs a measurement $x_i$ and obtains an outcome $a_i$. The coefficients $c_{a_1a_2}^{x_1x_2}$ are real numbers and $P(a_1a_2|x_1x_2)$ represents the observed outcome distribution, for each measurement pair. 
A \textit{lifting} of this Bell inequality to $n$ observers consists in extending the expression \eqref{Bellineq2} by choosing a fixed measurement and outcome for observers $A_3,\ldots,A_n$:
\begin{equation}\label{liftineq}
\mathcal{I}_{\vec0|\vec0}^{A_1A_2}=\sum_{a_1a_2x_1x_2}c_{a_1a_2}^{x_1x_2}P(a_1a_2\vec0|x_1x_2\vec0)\leq 0
\end{equation}
where, without loss of generality, the fixed $n-2$ measurements and outcomes are set to $\vec0=\{0,\ldots,0\}$. Notice that $P(a_1a_2\vec0|x_1x_2\vec0)=P(a_1a_2|x_1x_2,\vec0,\vec0)P(\vec0|\vec0)$, where $P(\vec0|\vec0)$ is independent of measurements $x_1$ and $x_2$ according to the no-signalling principle. This means that if the conditional distribution $P_{\vec0,\vec0}(a_1a_2|x_1x_2)\equiv P(a_1a_2|x_1x_2,\vec0,\vec0)$ violates the bipartite inequality \eqref{Bellineq2}, it implies that the full distribution $P(\vec a|\vec x)$ violates the lifted inequality \eqref{liftineq}. Therefore, the nonlocality of the conditional distribution is a sufficient condition for the nonlocality of the full distribution.\\

We now want to show that any biseparable distribution \eqref{biseparable} where parties $A_1$ and $A_2$ belong to different groups of parties, $A_1\in g$ and $A_2\in \bar g$,  does not violate a lifted Bell inequality \eqref{liftineq}:
\begin{equation}
\mathcal{I}_{\vec0|\vec0}^{A_1A_2}(P_\textrm{bisep}^{g/\bar g})\leq0
\end{equation}
Since a Bell inequality is a linear function and $P_\textrm{bisep}^{g/\bar g}$ is convex, it is enough to show that the previous inequality holds for any pure biseparable distribution $P(\vec a_g|\vec x_g)P(\vec a_{\bar g}|\vec x_{\bar g})$. We have then 
\begin{equation}\begin{split}
\mathcal{I}_{\vec0|\vec0}^{A_1A_2}\big(P(\vec a_g|\vec x_g)P(\vec a_{\bar g}|\vec x_{\bar g})\big)=\sum_{a_1a_2x_1x_2}c_{a_1a_2}^{x_1x_2}P_g(a_1\vec0|x_1\vec0)P_{\bar g}(a_2\vec0|x_2\vec0)\\
=\sum_{a_1a_2x_1x_2}c_{a_1a_2}^{x_1x_2}P_{A_1}(a_1|x_1,\vec0,\vec0)P_{A_2}(a_2|x_2,\vec0,\vec0)P_{g\setminus{A_1}}(\vec0|\vec0)P_{\bar g\setminus{A_1}}(\vec0|\vec0)\leq0
\end{split}
\end{equation}
where the size of the vector $\vec0$ should be clear by the context. Notice that we have again used the fact that the distributions are no-signalling and that $P_{A_i}(a_i|x_i,\vec0,\vec0)$ are well-defined local distributions.

\section{Proof of Theorem \ref{GG:THRM2}}\label{APPGG:TheoremN3}

We here provide a formal proof of Theorem \ref{GG:THRM2}. Let us start with an observation that will be used in the upcoming proof.\\

\textbf{Observation}\label{GG:OBS}.-- On any pure, non-maximally entangled, two-qubit state
\begin{equation}\label{GG:qubit}
\ket{\phi_{\theta}} = \cos(\theta)\ket{00} + \sin(\theta)\ket{11}
\end{equation}
i.e. for $\theta \in ]0,\frac{\pi}{4}[$, the measurements
\begin{equation} \label{HardyMeas}
\begin{split}
M_{0|0} = \cos(\alpha)\bra{0} + e^{i \delta} \sin(\alpha)\bra{1} \\
M_{0|1} \propto \cos^2(\theta)\cos(\alpha)\bra{0} + e^{i \delta}\sin^2(\theta)\sin(\alpha)\bra{1} \\ \\
N_{0|0} \propto \sin^3(\theta)e^{i \delta}\sin(\alpha)\bra{0} - \cos^3(\theta)\cos(\alpha)\bra{1} \\
N_{0|1} \propto \sin(\theta)\sin(\alpha)e^{i \delta}\bra{0} -  \cos(\theta)\cos(\alpha)\bra{1}
\end{split}
\end{equation}
lead to correlations $P_{\theta}(ab|xy) =\bra{\phi_{\theta}} \big( M_{a|x} \otimes N_{b|y} \big) \ket{\phi_{\theta}}$ that violate inequality \eqref{IA1A2} $I^{A_1A_2}(P_{\theta}(ab|xy)) > 0$ with the free parameters $\alpha$ and $\delta$ such that $\alpha \neq 0,\pi/2$ for any $\theta \neq 0,\pi/4$. More precisely, they lead to the particular violation of the inequality \eqref{IA1A2}
\begin{equation} \label{GG:QubitsResult}
I^{A_1A_2}(P_{\theta}(ab|xy)) = P_{\theta}(00|00) > 0
\end{equation}
and thus $P_{\theta}(01|01) = P_{\theta}(10|10) = P_{\theta}(00|11) = 0$, i.e. a realisation of the bipartite Hardy paradox \cite{Hardy}.  A proof of this observations can be found further in the appendix \ref{GGAPP:HardyMeasurements}.

Since we are interested in a violation up to any extent of our inequality 
\begin{align}
I^{A_1A_2}(P_{\theta}(ab|xy)) > 0
\end{align}
whose bound is zero, we have taken the freedom not to normalise some of the measurements in \eqref{HardyMeas}. In other words, the observation \eqref{GG:OBS} implies that for any non-maximally entangled pure two qubit state $\ket{\phi_{\theta}}$ \eqref{GG:qubit}, one can chose one of the measurement of one of the parties for free (as expressed by the free parameters $\alpha$ and $\delta$ such that $\alpha \neq 0,\pi/2$) and still find three other measurements such that the generated correlations violated the inequality.\\

Now we want to show that a large class of three qubit GME states violate inequality$I^{A_1A_2A_3}_{\mu = 0}$ \eqref{IA1A2A3}. In \cite{Acin2000}, it was shown that all three qubits in a pure state could be written as
\begin{equation}\begin{split}\label{3qubits}
\ket{\Psi_3} = h_0 \ket{000} + h_1 e^{i \phi} \ket{100} + h_2 \ket{101} + h_3 \ket{110} + h_4 \ket{111}
\end{split}\end{equation}
where $h_i \geq 0$, $\sum\limits_i h_i^2 = 1$ and $\phi \in [0,\pi]$. On these states, we impose the additional constrain that $h_2 = h_3$, i.e. we consider only the states \eqref{3qubits} which are symmetrical with respect to the permutations of the parties $A_2 \leftrightarrow A_3$. By  relabelling the parties' index, however, any state which is symmetrical with respect to the permutation of two out of the three parties can be transformed to one where the symmetry is between parties $A_2$ and $A_3$, which we chose without loss of generality. Now, party $A_2$ and $A_3$ both make the same projective measurement $\bra{m_{a_i|x_i}}$ for their input choice $x_2 = x_3 = 0$
\begin{equation}\begin{split}\label{3qubitsMeasBC}
\bra{m_{0|x_i=0}} =  \cos (\alpha) \bra{0} + \sin (\alpha)  \bra{1}
\end{split}\end{equation}
for some (yet) free angle $\alpha$\footnote{Remark additionally that we do not make use of a potential second degree of freedom (the phase).}. The state that is prepared  between parties $A_1A_3$ (resp. $A_1A_2$) from party $A_2$ ($A_3$) by performing measurement $\bra{m_{0|x_i=0}}$ \eqref{3qubitsMeasBC} on the state $\ket{\Psi_3}$ \eqref{3qubits} conditioned on obtaining the outcome $a_2 = 0$ ($a_3 = 0$) is
\begin{equation}\begin{split}\label{postmeasuMeasBC}
\ket{\psi_{0|0}^{A_1A_2}} = \ket{\psi_{0|0}^{A_1A_3}} \propto \cos (\alpha) h_0 \ket{00} + \big( \cos (\alpha) h_1 + \sin (\alpha) h_2 \big) \ket{10} \\
+ \big( \cos (\alpha) h_2 + \sin (\alpha) h_4 \big) \ket{11}
\end{split}\end{equation}
since $h_2 = h_3$ and that both the state $\ket{\Psi_3}$ \eqref{3qubits} and measurements $\bra{m_{0|x_i}}$ are symmetrical with respect to permutation $A_2 \leftrightarrow A_3$. Using the concurrence, the state $\ket{\psi_{0|0}^{A_1A_2}}$ \eqref{postmeasuMeasBC} is entangled if and only if
\begin{equation}\begin{split}
\det \big(\begin{tabular}{ c c }  $ \cos (\alpha) h_0$ & $0$ \\  $ \cos (\alpha) h_1 + \sin (\alpha) h_2$ & $\cos (\alpha) h_2 + \sin (\alpha) h_4$  \end{tabular} \big) \neq 0 \\ \vspace{0.3cm}
\Leftrightarrow  \cos (\alpha) h_0 \big( \cos (\alpha) h_2 + \sin (\alpha) h_4 \big) \neq 0
\end{split}\end{equation}
leading to the four conditions
\begin{equation}\begin{split}\label{GG:ThrmCond1}
\alpha \neq \frac{\pi}{2}
\end{split}\end{equation}
\begin{equation}\begin{split}\label{GG:ThrmCond2}
\tan(\alpha) \neq -\frac{h_2}{h_4}
\end{split}\end{equation}
\begin{equation}\begin{split}\label{GG:ThrmCond3}
h_0 \neq 0
\end{split}\end{equation}
\begin{equation}\begin{split}\label{GG:ThrmCond4}
h_2 \neq 0 \neq h_4
\end{split}\end{equation}
First, remark that both conditions \eqref{GG:ThrmCond3} and \eqref{GG:ThrmCond4} only mean that the state $\ket{\Psi_3}$ \eqref{3qubits} needs to be GME (as well as symmetrical $h_2 = h_3$). Now, since the parameter $\alpha$ is free, we choose to avoid the two values $\alpha = \frac{\pi}{2}$ and $\alpha = -\arctan \big( \frac{h_2}{h_4} \big) \neq 0$. In the end, one can tune continuously the parameter $\alpha$ (up to the forbidden values \eqref{GG:ThrmCond1} and \eqref{GG:ThrmCond2}) so that the prepared states $\ket{\psi_{0|0}^{A_1A_2}} = \ket{\psi_{0|0}^{A_1A_3}}$ are not maximally entangled. One can then use observation \ref{GG:OBS}, as well as the symmetries $A_2 \leftrightarrow A_3$ that was imposed on both state and measurements, to obtain
\begin{equation}\begin{split}\label{GG:ThrmFINAL}
I^{A_1A_2A_3}_{\mu = 0} = I^{A_1A_2}_{0|0}+I^{A_1A_3}_{0|0}-P(000|000) = 2 I^{A_1A_2}_{0|0} - P(000|000)\\ = P(000|000)-2P(100|100)-2P(010|010)-2P(000|110) > 0
\end{split}\end{equation}
by choosing $A_1$'s measurements as in \eqref{HardyMeas} for the prepared (non maximally entangled) state $\ket{\psi_{0|0}^{A_1A_2}}$ \eqref{postmeasuMeasBC}, i.e. realising
\begin{equation}\begin{split} \label{GG:QubitsResult3part}
P(000|000) > 0\\
P(010|010) = 0\\
P(100|100) = 0\\
P(000|110) = 0
\end{split}\end{equation}
%Note, however, that we additionally need $(P_{A_2}(0|0)=)$ $P_{A_3}(0|0) > 0$, so 

%One will note the strong connection of our results with 

%By using a measure, such as the concurrence, of the entanglement in the states \eqref{postmeasuMeasBC} one gets that the prepared states are entangled as long as

\section{Hardy's measurements for $n=2$}\label{GGAPP:HardyMeasurements}

From the realisation \eqref{GG:QubitsResult}, we have four conditions
\begin{equation}\begin{split} \label{GG:QubitsResult2part}
P(00|00) > 0\\
P(01|01) = 0\\
P(10|10) = 0\\
P(00|11) = 0
\end{split}\end{equation}
to be satisfied by the measurement $M_{a|x}$ and $N_{b|y}$ made on the state $\ket{\phi_{\theta}} = \cos(\theta)\ket{00} + \sin(\theta)\ket{11}$ written in it's Schmidt basis by A and B respectively. We start by choosing $M_{0|0} = \cos\alpha \bra{0} + \sin\alpha e^{i \delta} \bra{1}$ freely and then try to satisfy these four conditions. From $P(01|01) = 0$ we get that\\

\begin{align}
(\cos\alpha \bra{0} + \sin\alpha e^{i \delta} \bra{1}) \otimes N_{0|1} \cdot (\cos(\theta)\ket{00} + \sin(\theta)\ket{11}) = 0 \nonumber \\  
\Leftrightarrow N_{0|1} (\cos\alpha \cos\theta \ket{0} + e^{i \delta} \sin\alpha \sin\theta \ket{1}) = 0 \nonumber \\
\Leftrightarrow N_{0|1} \propto e^{i \delta}\sin\alpha \sin\theta \bra{0} - \cos\alpha \cos\theta \bra{1}
\end{align}\\

where we use non normalized measurements, which, again, does not make a difference when interested in conditions of the form $P(a_1a_2|x_1x_2)=0$ or $P(a_1a_2|x_1x_2)>0$. Considering projective two-outcome measurements: \\

\begin{align}
N_{1|1} \propto \cos\alpha \cos\theta \bra{0} + e^{-i \delta} \sin\alpha \sin\theta  \bra{1}
\end{align}

Then, with condition $P(00|11) = 0$

\begin{align}
M_{1|1} \otimes N_{1|1} (\cos(\theta)\ket{00} + \sin(\theta)\ket{11}) = 0 \nonumber \\
\Leftrightarrow M_{1|1} (\cos\alpha \cos^2\theta \ket{0} + e^{-i \delta} \sin\alpha \sin^2\theta \ket{1}) = 0 \nonumber \\
\Rightarrow M_{1|1} \propto e^{-i \delta} \sin\alpha \sin^2\theta \bra{0} - \cos\alpha \cos^2\theta \bra{1} \\
\Rightarrow M_{0|1}  \propto \cos\alpha \cos^2\theta \bra{0} + e^{i \delta} \sin\alpha \sin^2\theta \bra{1}
\end{align}

Finally, from condition $P(10|10) = 0$

\begin{align}
M_{0|1} \otimes N_{0|0} (\cos(\theta)\ket{00} + \sin(\theta)\ket{11}) = 0 \nonumber \\
\Rightarrow N_{0|0} \propto e^{i \delta} \sin\alpha \sin^3\theta \bra{0} - \cos\alpha \cos^3\theta \bra{1} \\
\Rightarrow N_{1|0} \propto \cos\alpha \cos^3\theta \bra{0} +  e^{-i \delta} \sin\alpha \sin^3\theta \bra{1}
\end{align}

Now one can check that with these measurements on the state $\cos(\theta)\ket{00} + \sin(\theta)\ket{11}$ gives:

\begin{align}
M_{0|0} \otimes N_{0|0} (\cos(\theta)\ket{00} + \sin(\theta)\ket{11}) \propto ... = -\frac{e^{i \delta}}{8} \sin 2\alpha \sin 4\theta
\end{align}

That is equal to zero -- i.e. $P(00|00) = 0$ -- if and only if $\alpha = 0, \pi/2$ or $\theta = 0,\pi/4$. In the end, the conditions \eqref{GG:QubitsResult2part} are satisfied for these measurements for all non-maximally entangled states with any set of measurements of the form
\begin{equation}
\begin{split}
M_{0|0} = \cos(\alpha)\bra{0} + e^{i \delta} \sin(\alpha)\bra{1} \\
M_{0|1} \propto \cos^2(\theta)\cos(\alpha)\bra{0} + e^{i \delta}\sin^2(\theta)\sin(\alpha)\bra{1} \\ \\
N_{0|0} \propto \sin^3(\theta)e^{i \delta}\sin(\alpha)\bra{0} - \cos^3(\theta)\cos(\alpha)\bra{1} \\
N_{0|1} \propto \sin(\theta)\sin(\alpha)e^{i \delta}\bra{0} -  \cos(\theta)\cos(\alpha)\bra{1}
\end{split}
\end{equation}
except for the forbidden values of $\alpha = 0, \pi/2$.

\section{Properties of our families of Bell inequalities}
\subsection{The family of Bell inequalities $I_{\textrm{sym}}^{A_1A_2...A_n}$ \eqref{IGeneralsym} witnesses genuine multipartite nonlocality}\label{APP2:IGeneral}
In this section we want to give a more detailed proof of Theorem \ref{NpartyisGMNL}, which states that 
for any number $n\geq3$ of observers, all biseparable distributions \eqref{biseparable} satisfy our family of inequalities \eqref{IGeneralsym},
\begin{equation}\begin{split}\label{app:IGeneral}
I_{\textrm{sym}}^{A_1A_2...A_n} =\sum\limits_{i=1}^{n-1}\sum\limits_{j>i}^{n} I^{A_iA_j}_{\vec{0}|\vec{0}} - {n-1\choose 2}  P(\vec{0}|\vec{0}) \leq 0
\end{split}\end{equation}
where ${n-1\choose 2}  = \frac{(n-1)(n-2)}{2}$ and thus $I^{A_1A_2...A_n}$ witnesses GMNL in the distributions. The proof for the family of inequalities $I^{A_1\ldots A_n}_{\textrm{\ding{192}}}$ \eqref{IGeneral} follows exactly the same lines.

\begin{proof} Our Bell inequalities $I^{A_1A_2...A_n}$  are invariant under permutations of the observers. Since a Bell inequality is a linear function of the probability terms $P(\vec{a}|\vec{x})$, and by the convexity of biseparable distributions \eqref{biseparable}, we can restrict the proof -- without loss of generality -- to pure biseparable distributions of the form 
\begin{equation}\begin{split}\label{Bisepn}
P_{m/(n-m)} \equiv P(a_1a_2\ldots a_m|x_1x_2\ldots x_m)P(a_{m+1}a_{m+2}\ldots a_n|x_{m+1}x_{m+2}\ldots x_n)\,,
\end{split}\end{equation}
where the first term includes the variables of the $m$ first observers and the second the remaining $n-m$. Let us recall that, inside each group, observers are allowed to share any no-signalling nonlocal resources. Our proof consists in counting how many lifted inequalities $I^{A_iA_j}_{\vec{0}|\vec{0}}$ \eqref{InLiftee} can be violated by a pure biseparable distribution \eqref{Bisepn}. We will see that this happens to \textit{at most} ${n-1\choose 2}$ lifted inequalities. Indeed, a term $I^{A_iA_j}_{\vec{0}|\vec{0}}$ can only be positive if observers $A_i$ and $A_j$ belong to the same group ($i,j \leq m$ or $i,j > m$), since otherwise there are only classically correlated (see Appendix~\ref{APP1:liftings}). Thus
%\begin{widetext}
\begin{equation}\begin{split}\label{ToProveN1}
I_{\textrm{sym}}^{A_1A_2...A_n}(P_{m/(n-m)}) \leq \sum\limits_{i=1}^{m-1}\sum\limits_{j>i}^{m} I^{A_iA_j}_{\vec{0}|\vec{0}} + \sum\limits_{k=m+1}^{n-1}\sum\limits_{l>k}^{n} I^{A_kA_l}_{\vec{0}|\vec{0}} - {n-1\choose 2} P(\vec{0}|\vec{0})\\
= \sum\limits_{i=1}^{m-1}\sum\limits_{j>i}^{m} I^{A_iA_j}_{\vec{0}|\vec{0}}- {m\choose 2}P(\vec{0}|\vec{0}) + \sum\limits_{k=m+1}^{n-1}\sum\limits_{l>k}^{n} I^{A_kA_l}_{\vec{0}|\vec{0}} - {n-m\choose 2}P(\vec{0}|\vec{0})\\ - (m-1)(n-m-1) P(\vec{0}|\vec{0})\\
=  \sum\limits_{i=1}^{m-1}\sum\limits_{j>i}^{m} \bar{I}^{A_iA_j}_{\vec{0}|\vec{0}} + \sum\limits_{k=m+1}^{n-1}\sum\limits_{l>k}^{n} \bar{I}^{A_kA_l}_{\vec{0}|\vec{0}} - (m-1)(n-m-1) P(\vec{0}|\vec{0})\\ \leq - (m-1)(n-m-1) P(\vec{0}|\vec{0}) \leq 0
\end{split}
\end{equation}
%\end{widetext}
in which we have used the fact that  $\sum\limits_{i=1}^{m-1}\sum\limits_{j>i}^{m} I^{A_iA_j}_{\vec{0}|\vec{0}}$ contains ${m\choose 2}$ lifted terms and $ \sum\limits_{k=m+1}^{n-1}\sum\limits_{l>k}^{n} I^{A_kA_l}_{\vec{0}|\vec{0}}$ contains ${n-m\choose 2}$ of them. We have further used  $\bar{I}^{A_iA_j}_{\vec{0}|\vec{0}} \equiv I^{A_iA_j}_{\vec{0}|\vec{0}} - P(\vec{0}|\vec{0}) \leq 0$, for any $i,j$ (see equation \eqref{Ibar}). Notice that the situation where most lifted terms  $I^{A_iA_j}_{\vec{0}|\vec{0}}$ could be positive occurs for bipartitions of one versus $n-1$ observers, hence the ${n-1\choose 2}$ factor in our Bell inequalities \eqref{app:IGeneral}.
%Finally, we have also used the relation
%
%\begin{equation}\begin{split}\label{PropCombi}
%{n-1\choose 2} = {m\choose 2}+{n-m\choose 2} +(m-1)(n-m-1)
%\end{split}\end{equation}
%
%to develop the term ${n-1\choose 2} P(\vec{0}|\vec{0})$ and the fact that $m<n$.
\end{proof}

\subsection{A recursive formula for our inequalities} \label{APP:Alternative}

Our family of Bell inequalities $I_{\textrm{sym}}^{A_1A_2...A_n} $ can also be written in a recursive form, which shows its rich multipartite structure and operational meaning:
\begin{equation}\begin{split}\label{app:IGeneral2}
I_{\textrm{sym}}^{A_1A_2...A_n} = \frac{1}{n-2}\sum\limits_{i = 1}^{n} I^{\textrm{all}\setminus A_i}_{0|0} - P(\vec{0}|\vec{0}) \leq 0
\end{split}\end{equation}
for $n\geq3$, where $I^{\textrm{all}\setminus A_i}_{0|0}$ is the Bell inequality testing genuine nonlocality between $n-1$ parties lifted to $n$ parties, with party $A_i$'s input and outcome set to $0$
\begin{equation}\label{IGeneral3}
I^{\textrm{all}\setminus A_i}_{0|0} = \frac{1}{n-3}\sum\limits_{\substack{j=1 \\ j \neq i}}^{n} I^{\textrm{all}\setminus A_iA_j}_{0|0} - P(\vec{0}|\vec{0})\,.
\end{equation}
The seed of this recursive expression is the variant of the CHSH inequality \eqref{IA1A2}.

\begin{proof}
We prove that the recursive expression \eqref{app:IGeneral2} is equivalent to the direct expression \eqref{app:IGeneral} for $I_{\textrm{sym}}^{A_1A_2...A_n}$ through mathematical induction. First, we check that for $n=3$ the equivalence holds, which can easily be done by developing both expressions. Then, we show that if the equivalence is true for $n$, it implies that it is true also for $n+1$.

Suppose the equivalence holds for $n$:
\begin{equation}\begin{split}\label{EquProof1}
 \frac{1}{n-2}\sum\limits_{i = 1}^{n} I^{\textrm{all}\setminus A_i}_{0|0} - P(\vec{0}|\vec{0}) = \sum\limits_{i=1}^{n-1}\sum\limits_{j>i}^{n} I^{A_iA_j}_{\vec{0}|\vec{0}} - {n-1\choose 2}  P(\vec{0}|\vec{0}) \,.
\end{split}\end{equation}

For $n+1$, we develop the recursive expression in \eqref{app:IGeneral2}, where $I^{\textrm{all}\setminus A_i}_{0|0}$ is now an $n$ observer inequality for which the recurrence hypothesis \eqref{EquProof1} can be used:
%\begin{equation}\begin{split}\label{EquProof2}
%I^{A_1A_2...A_nA_{n+1}} = \frac{1}{n-1}\sum\limits_{i = 1}^{n+1} I^{\textrm{all}\setminus A_i}_{0|0} - P(\vec{0}|\vec{0})
%\end{split}\end{equation}
\begin{equation}\begin{split}\label{EquProof3}
\frac{1}{n-1}\sum\limits_{i = 1}^{n+1} I^{\textrm{all}\setminus A_i}_{0|0} - P(\vec{0}|\vec{0})\\ =^{\eqref{EquProof1}} \frac{1}{n-1}\sum\limits_{i = 1}^{n+1}\bigg( \sum\limits_{\substack{j=1 \\ j \neq i}}^{n}&\sum\limits_{\substack{k>j \\ k \neq  i}}^{n+1} I^{A_jA_k}_{\vec{0}|\vec{0}} - {n-1\choose 2}  P(\vec{0}|\vec{0}) \bigg)  - P(\vec{0}|\vec{0}) \\
= \frac{1}{n-1}\sum\limits_{i = 1}^{n+1}\bigg( \sum\limits_{\substack{j=1 \\ j \neq i}}^{n}\sum\limits_{\substack{k>j \\ k \neq i}}^{n+1} &I^{A_jA_k}_{\vec{0}|\vec{0}} \bigg)  -  \bigg(  \frac{n+1}{n-1} {n-1\choose 2} + 1 \bigg)P(\vec{0}|\vec{0})
\end{split}\end{equation}
%Finally, we develop further 
%\begin{equation}\begin{split}\label{EquProof4}
%\frac{n+1}{n-1} {n-1\choose 2} + 1 = \frac{n+1}{n-1}\frac{(n-1)(n-2)}{2}+1 = \frac{(n+1)(n-2)}{2}+\frac{2}{2} = \frac{n^2-n}{2} = \frac{n(n-1)}{2} = {n\choose 2}
%\end{split}\end{equation}
Note that the last expression can be simplified taking into account that the terms $I^{A_jA_k}_{\vec{0}|\vec{0}}$ are being counted multiple times.  Since the inequalities are invariant under permutations of observers, we can restrict our attention to counting how many times the particular term $I^{A_1A_2}_{\vec{0}|\vec{0}}$ appears in \eqref{EquProof3}. %In the term $\frac{1}{n-1}\sum\limits_{i = 1}^{n+1}\bigg( \sum\limits_{\substack{j=1 \\ j \neq i}}^{n}\sum\limits_{\substack{k>j \\ k \neq i}}^{n+1} I^{A_jA_k}_{\vec{0}|\vec{0}} \bigg)$, 
One can check that $\sum\limits_{\substack{j=1 \\ j \neq i}}^{n}\sum\limits_{\substack{k>j \\ k \neq i}}^{n+1} I^{A_jA_k}_{\vec{0}|\vec{0}}$ gives one term $I^{A_1A_2}_{\vec{0}|\vec{0}}$ if $i \neq 1,2$. Suming over $i$, we get a total of $n-1$ terms, from which we obtain % that cancels the factor $\frac{1}{n-1}$:
\begin{equation}\begin{split}\label{EquProof5}
\frac{1}{n-1}\sum\limits_{i = 1}^{n+1}\bigg( \sum\limits_{\substack{j=1 \\ j \neq i}}^{n}\sum\limits_{\substack{k>j \\ k \neq i}}^{n+1} I^{A_jA_k}_{\vec{0}|\vec{0}} \bigg)  -  \bigg(  \frac{n+1}{n-1} {n-1\choose 2} + 1 \bigg)P(\vec{0}|\vec{0}) = \sum\limits_{j=1}^{n}\sum\limits_{k>j}^{n+1} I^{A_jA_k}_{\vec{0}|\vec{0}}  - {n\choose 2} P(\vec{0}|\vec{0})
\end{split}\end{equation}
where we used $\frac{n+1}{n-1} {n-1\choose 2} + 1={n\choose 2}$. Since the last expression coincides with the direct expression \eqref{app:IGeneral} for $n+1$ observers, we finish our proof.
\end{proof}
%
% This alternative way of writing our inequalities also has a clear meaning: any GMNL distribution violating our inequalities for $n$ observers also has all marginal conditional distributions of $m$ observers conditioned the remaining observers making measurement $x_i=0$ and obtaining outcome $a_i=0$. 

%As in the case of \eqref{IbarA1A2}
%\begin{equation}\begin{split}\label{IbarA1A2A3}
%\bar{I}^{A_1A_2}_{0|0} \equiv I^{A_1A_2}_{0|0} - P(000|000)\\ = -P(010|010)-P(100|100)-P(000|110) \leq 0
%\end{split}\end{equation}
%for any valid probability distribution (it is not a Bell inequality).

%In this section of the appendices, we develop further our study of the family of Bell inequalities
%\begin{equation}\begin{split}\label{app:IGeneral}
%I^{A_1A_2...A_n} =\sum\limits_{i=1}^{n-1}\sum\limits_{j>i}^{n} I^{A_iA_j}_{\vec{0}|\vec{0}} - {n-1\choose 2}  P(\vec{0}|\vec{0}) \leq 0
%\end{split}\end{equation}
%where ${n-1\choose 2}  = \frac{(n-1)(n-2)}{2}$ and 
%\begin{equation}\begin{split}\label{app:InLiftee}
%I^{A_iA_j}_{\vec{0}|\vec{0}} = P(0_i0_j\vec{0}|0_i0_j\vec{0})-P(1_i0_j\vec{0}|1_i0_j\vec{0}) -P(0_i1_j\vec{0}|0_i1_j\vec{0})-P(0_i0_j\vec{0}|1_i1_j\vec{0})
%\end{split}\end{equation}

\subsection{Fully local strategies that saturate the inequalities}

Interestingly, one can check that the (fully) local strategy
\begin{equation}\begin{split}\label{LocalStrategy}
P_{\textrm{L}}(a_1a_2...a_n|x_1x_2...x_n)
= \begin{cases} 1 \hspace{0.3cm} \textrm{if} \hspace{0.3cm} a_i = 1 \hspace{0.2cm} \forall i \hspace{0.1cm}\textrm{and} \hspace{0.1cm} \forall x_i\\
0 \hspace{0.3cm} \textrm{else} \end{cases}
\end{split}\end{equation}
saturates our families of inequalities \eqref{IGeneralsym} and \eqref{IGeneral} since there is no term in the inequalities where all outcomes have value $1$. Nonlocal resources shared between a subset of the observers are thus useless to reach better bounds on our family, only nonlocal resources shared between \textit{all} observers are relevant. Remark that these observation generalise to all the families of inequalities that have the CHSH inequality $I^{A_1A_2}$ \eqref{IA1A2} as seed.

\subsection{Post-quantum no-signalling resources that violate the inequalites}
Consider a genuine multipartite generalisation of the (no-signalling) PR-box \cite{NS}:
\begin{equation}\label{NSbox}
P_{\textrm{NS}}(\vec a|\vec x) = \begin{cases}
\frac{1}{2^{n-1}} \hspace{0.2cm} \textrm{if} \hspace{0.2cm} \oplus_{i = 1}^{n} a_i = \oplus_{i = 1}^{n-1}\oplus_{j>i}^{n} x_ix_j\\
0  \hspace{0.7cm} \textrm{else}
\end{cases}
\end{equation}
where the marginal distributions are completely random, i.e. $P_{\textrm{NS}}(a_i|x_i)=\frac{1}{2}$,$\forall i$. It is interesting to see that this post-quantum no-signalling distribution violates our Bell inequalities $I^{A_1A_2...A_n}$, for all $n\geq2$, 

\begin{equation}
I_{\textrm{sym}}^{A_1A_2...A_n}(P_{\textrm{NS}}) = \frac{n-1}{2^{n-1}} >0\,.
\end{equation}

\begin{proof} The proof follows from direct evaluation of our inequalities \eqref{app:IGeneral} with the no-signalling box \eqref{NSbox}. First, we get that $I_{\vec{0}|\vec{0}}^{A_iA_j}(P_{\textrm{NS}}) =^{\eqref{InLiftee}} P_{\textrm{NS}}(\vec0|\vec0) \hspace{0.3cm} \forall i,j$ because $P_{\textrm{NS}}(\vec0|\vec0)$  is the only non-vanishing term. Then 
\begin{equation}\begin{split}
I_{\textrm{sym}}^{A_1A_2...A_n}(P_{\textrm{NS}}) =  \sum\limits_{i = 1}^{n-1}\sum\limits_{j>i}^n P_{\textrm{NS}}(\vec0|\vec0) -{n-1\choose2} P_{\textrm{NS}}(\vec{0}|\vec{0})\\= \frac{1}{2^{n-1}}\Big[{n \choose 2}-{n-1\choose 2}\Big] = \frac{n-1}{2^{n-1}} > 0, \hspace{0.3cm} \forall n\geq2
\end{split}\end{equation}
which finishes our proof. \end{proof}

A similar proof can be made for the inequalities in the family $I^{A_1\ldots A_n}_{\textrm{\ding{192}}}$ \eqref{IGeneral}.

\section{All pure GME states of the family \textbf{$\ket{GHZ^n}_{\theta} = \cos\theta\ket{0}^{\otimes n} - \sin\theta\ket{1}^{\otimes n}$} generate GMNL correlations}\label{APP3:GME-GMNL}
Here we prove Theorem \ref{GG:THRM1} in detail.

\begin{proof}
Our proof is constructive as we will provide, for all states
\begin{equation}\begin{split}\label{GHZapp}
\ket{GHZ^n}_{\theta} = \cos\theta\ket{0}^{\otimes n} - \sin\theta\ket{1}^{\otimes n}
\end{split}\end{equation}
with $\theta \in ]0,\frac{\pi}{4}[$, local measurements that lead to explicit distributions $P_{GHZ^n_\theta}(\vec{a}|\vec{x})$ violating our family of inequalities $I_{\textrm{sym}}^{A_1...A_n}$\eqref{IGeneralsym}. 

%For simplicity of the proof we will work, without loss of generality, in the unconventional range $\theta \in ]-\frac{\pi}{4},0[$. 
In order to provide symmetry to the problem, and significantly reduce the degrees of freedom,  all the observers use the same projective measurements $m_{a_i|x_i}= m_{a|x}$:
\begin{equation}\begin{split}\label{Measurements}
m_{0|x} &= \cos\alpha_x\bra{0}+\sin\alpha_x\bra{1}\\
m_{1|x} &= \sin(\alpha_x\bra{0}-\cos\alpha_x\bra{1}\,.
\end{split}\end{equation}
Since both the state \eqref{GHZapp} and measurements \eqref{Measurements} are invariant under permutations of the observers, the generated distribution $P_{GHZ^n_\theta}(\vec{a}|\vec{x})$ also has this symmetry. For three observers for example, we get that $P(100|100) = P(010|010) = P(001|001)$ or $P(000|011) = P(000|101) = P(000|110)$ or that the lifted inequalities are all equal $I^{A_1A_2}_{\vec{0}|\vec{0}} = I^{A_1A_3}_{\vec{0}|\vec{0}} = I^{A_2A_3}_{\vec{0}|\vec{0}}$. This implies that inequalities $I^{A_1...A_n}$ \eqref{IGeneral}, when evaluated on the generated distributions, simplify to
\begin{equation}\begin{split}\label{Thrm1}
I_{\textrm{sym}}^{A_1A_2...A_n}\big(P_{GHZ_\theta^n}\big)= {n\choose 2} I^{A_1A_2}_{\vec{0}|\vec{0}} - {n-1\choose 2}P(\vec{0}|\vec{0})
\\= (n-1)P(00\vec{0}|00\vec{0}) - 2{n\choose 2}P(10\vec{0}|10\vec{0})-{n\choose 2}P(00\vec{0}|11\vec{0})
\end{split}\end{equation}
where we have used that ${n\choose 2}-{n-1\choose 2} = n-1$.
Using measurements \eqref{Measurements} on the state \eqref{GHZapp} we obtain all the terms of \eqref{Thrm1}
\begin{equation}\begin{split}\label{Thrm3}
P(00\vec{0}|00\vec{0}) &= \big(\cos^n (\alpha_0)\cos(\theta) - \sin^n(\alpha_0)\sin(\theta)\big)^2\\
P(10\vec{0}|10\vec{0}) &= \big(\cos^{n-1}(\alpha_0)\sin(\alpha_1)\cos(\theta) + \sin^{n-1} (\alpha_0)\cos(\alpha_1)\sin(\theta)\big)^2\\
P(00\vec{0}|11\vec{0})  &= \big(\cos^{n-2} (\alpha_0)\cos^2 (\alpha_1)\cos(\theta)- \sin^{n-2}(\alpha_0)\sin^2(\alpha_1)\sin(\theta)\big)^2
\end{split}\end{equation}

We want now to find angles $\alpha_x$ of the local measurements \eqref{Measurements} such that the quantity \eqref{Thrm1} is always positive. A particular solution is
\begin{equation}\begin{split}\label{Thrm2}\begin{cases}
P(00\vec{0}|00\vec{0}) > 0\\
P(10\vec{0}|10\vec{0}) = P(00\vec{0}|11\vec{0}) = 0\end{cases}\,.
\end{split}\end{equation}
which holds true for angles
\begin{equation}\begin{split}\label{angles}
\alpha_0 &= \textrm{arctan}(\tan^{\frac{-3}{3n-4}}\theta)\\
\alpha_1 &= -\textrm{arctan}(\tan^{\frac{-1}{3n-4}}\theta)
\end{split}\end{equation}
when $\theta \in ]0,\frac{\pi}{4}[$. The value of the inequalities at these angles is
\begin{equation}\begin{split}\label{IGHZ}
I_{\textrm{sym}}^{A_1A_2...A_n}(P_{GHZ}(\vec{a}|\vec{x})) = (n-1)P(00\vec{0}|00\vec{0})\\ =(n-1)\big(\cos^n (\textrm{arctan}(\tan^{\frac{-3}{3n-4}}\theta)\big)\cos(\theta) - \sin^n\big(\textrm{arctan}(\tan^{\frac{-3}{3n-4}}\theta))\sin(\theta)\big)^2
\end{split}\end{equation}
which is positive for $\theta\in]0,\pi/4[$, as promised. 
\end{proof}
For the maximally entangled state ($\theta=\pi/4$) we have $P(\vec{0}|\vec{0}) = 0$, which means that our construction breaks. However, this state is already known to be genuine multipartite nonlocal for all number of observers \cite{BancalQuantifying}, and moreover we numerically found several sets of measurements on it that lead to a violation of our inequalities. 
%\end{document}

\end{document}